\documentclass[dvipdfmx,12pt]{article}

\usepackage{fancybox}
\usepackage{cite}
\usepackage{float}
\usepackage{wrapfig}
\usepackage{amsfonts}
\usepackage{amsmath}
\usepackage{amsbsy}
\usepackage{amssymb}
\usepackage{amsthm}
\usepackage{bm}
\usepackage{epsfig}
\usepackage{latexsym}
\usepackage{color}
\usepackage{here}
\usepackage{graphicx}

\allowdisplaybreaks

\theoremstyle{plain}

\newtheorem{cor}{Corollary}[section]
\newtheorem{prop}{Proposition}[section]
\theoremstyle{definition}

\theoremstyle{remark}

\begin{document}
	\rightline{\baselineskip16pt\rm\vbox to20pt{
			{
				\hbox{OCU-PHYS-480}
				\hbox{AP-GR-145}
			}
			\vss}}%

	\bigskip
	\begin{center}
		{\LARGE\bf Black holes submerged in Anti-de Sitter space}\\
		\bigskip\bigskip
		{\large 
			Hideki ishihara\footnote{ishihara@sci.osaka-cu.ac.jp}, 
			Satsuki Matsuno\footnote{smatsuno@sci.osaka-cu.ac.jp}, and 
			Haruki Nakamura
		}\\
		\bigskip
		{\it Department of Physics, Graduate School of Science,
			Osaka City University\\
			3-3-138 Sugimoto, Sumiyoshi, Osaka 558-8585, Japan}
	\end{center}
	
	\begin{center}
		\today
	\end{center}
	
	\vspace{3cm}
	
	\begin{abstract}
		Suppose a one-dimensional isometry group acts on a space, we can consider a submergion 
		induced by the isometry, namely we obtain an orbit space by identification of points 
		on the orbit of the group action. We study the causal structure of the orbit space for 
		Anti-de Sitter space (AdS) explicitely. 
		In the case of AdS$_3$, we found a variety of black hole structure, 
		and in the case of AdS$_5$, we found a static four-dimensional black hole, 
		and a spacetime which has two-dimensional black hole as a submanifold.
	\end{abstract}

	\newpage
	
	\tableofcontents
	
	\newpage
	
	\section{Introduction}
	One considers that the space with rich geometrical symmetry is simple, 
	while the less symmetric space is complicate. 
	The maximally symmetric space is the most simple space in this sense. 
	In fact, our real universe has very complicate geometrical structure so that the universe, 
	where black holes exist and gravitational waves propagate, has interesting physical properties. 
	We discuss the possibility, in this paper, 
	that even though the space is simple, namely, it admits a rich isometry group, the space can involve 
	non-trivial geometrical structures from a viewpoint of decomposition of the space with respect 
	to the isometry.
	
	If an isometry group generated by a Killing vector acts on a space then one can decompose the space 
	into the orbit of the group action and the transverse space to the orbit. 
	Suppose that an $m$-dimensional Lie group G acts on an $(m+n)$-dimensional space $\tilde{M}$ transitively. 
	We obtain an $n$-dimensional manifold $M=\tilde{M}/G$, called the orbit space, by identifying 
	the points on a orbit of the action of G. 
	This procedure, $\pi:\tilde{M}\rightarrow M/G$, is called the \lq{\sl submersion}\rq\ 
	and it induces a fiber bundle structure, 
	where the fiber is a orbit of the group action and the base space is the orbit space.
	The submersion $\pi$ also defines the Riemannian submersion $\pi:\ (\tilde{M},\tilde{g})\rightarrow (M,g)$. 
	A tangent space $T_p\tilde{M}$ of $\tilde{M}$ is decomposed into the direct sum 
	\begin{align}
	T_p\tilde{M}=V_p\oplus H_p
	\end{align}
	where $V_p:=\ker\pi_\ast|_p$ is called the vertical subspace, 
	and $H_p$, called horizontal subspace, is the orthogonal complement to $V_p$ with respect to $\tilde{g}$. 
	The metric $g$ on $M$, called a projective metric, is obtained by restricting the metric $\tilde{g}$ on $\tilde M$  
	onto $H_p$.
	If the original space has a variety of isometries, depending on the choice of the isometry, 
	the orbit space can have complicated geometry even if the original space is simple. 
	
	Now, we suppose that $(\tilde{M},\tilde{g})$ is governed by the Einstein gravity.
	The Einstein gravity of $\tilde{M}$ is equivalent to the theory of gauge field with gauge group G 
	and adjoint scalar field of G, which interact with gravity in $n$-dimensional spacetime. 
	In the case of $\dim G=1$, as shown in the Appendix \ref{A brief review of the Kaluza-Klein theory 
		and the Einstein-Maxwell-Dilaton theory}, we obtain the Einstein-Maxwell-Dilaton theory that consists of 
	an abelian gauge field $A$ and an adjoint scalar field $\phi$ that interact with a gravitational field $g$. 
	Let $\xi$ be a Killing vector filed that generates an isometry group G. 
	The set of fields
	\begin{align}
	\hat{g}=|\xi|^\alpha g,\ A=\xi,\ e^\phi=|\xi|, 
	\label{decomposed_fields}
	\end{align}
	solves the Einstein-Maxwell-Dilaton theory. 
	If we chose the parameter $\alpha$ by
	\begin{align}
	\alpha=\frac{1}{n-2},
	\end{align}
	we have a set of the Einstein-Maxwell-Dilaton field in the Einstein frame. 
	Hereafter, we call this space $(M,\hat{g})$ the norm-weighted orbit space, or the base space, 
	simply in the term of fiber bundle. 
	
	In this paper, we consider the submersion by the isometry for the Anti-de Sitter space (AdS), 
	which gathers much attention recently in the context 
	of AdS/CFT\cite{Maldacena:1997re} or brane universe models\cite{Maartens:2010ar} etc. 
	Since AdS is maximally symmetric space, there exists a lot of Killing vectors which are not equivalent 
	geometrically as classified in \cite{Banados:1992gq, Holst:1997tm, Koike:2008fs, Morisawa:2017lpj}.
	We expect that the norm-weighted orbit space have a variety of geometric structures. 
	In particular we study the causal structure of base spacetimes $(M,\hat{g})$ of AdS. 
	
	In our study, 
	we develop a linear algebraic method to analyse the causal structure of the base space $(M,\hat{g})$ of AdS. 
	Since AdS$_n$ can be embedded in the pseudo Euclidean space $E^{(2,n-1)}$, 
	we can regard the Killing vector field used for the submersion as a linear operator in $E^{(2,n-1)}$. 
	We describe a property of null geodesics of the base space, which play an essential role for studying the causal 
	structure, in term of the linear operator in $E^{(2,n-1)}$. 
	This method enable us to analyze the behavior of null geodesics easily. 
	
	In the AdS$_3$ case, we clarify all causal structure of two-dimensional base spacetime,  
	where the fiber is one-dimension. 
	We found several kinds of causal structure that describe black holes. 	
	In AdS$_5$ case, we found a black hole spacetime and a spacetime which is not a complete four-dimensional black hole 
	but contains a two-dimensional black hole as a subspace. 
	
	Banados, Teitelboim and Zanelli\cite{BTZ} found that a quotient space of AdS by its some 
	discrete isometry has the black hole structure, called BTZ black hole. 
	Though the local geometry of BTZ spacetime is the same as AdS, the global structure is the black hole. 
	After their work, many researchers have studied BTZ black holes in 
	detail\cite{Banados:1992gq, Holst:1997tm, Aminneborg:2008sa}.
	The spacetime considered in the present paper would be obtained by the continuous limit of the BTZ spacetime, 
	the period of identification tends to 0, so it is expected that causal structure of submerged spacetime 
	is closely related to the BTZ spacetimes. 
	Indeed some of the conformal diagrams in the both the BTZ spacetime and our spacetime are the same.  
	However, the dimension of obtained spacetime and local geometry are different each other. 
	
	The organization of the paper is as follows. 
	An embedding of AdS is given in section 2, and the classification of the Killing vectors is 
	reviewed in section 3. The method of analysis to clarify the causal structure used in the present 
	paper is shown in section 4. In section 5 and section 6, the concrete analysis 
	for the AdS$_3$ and AdS$_5$ cases are done, respectively. Section 7 is devoted to conclusion and discussion.
	
	\section{Anti-de Sitter embedded in a flat space}
	\label{section Causal structure of AdS}
	Anti-de Sitter space (AdS) is the negative constant curvature space that can be defined as a submanifold 
	in a flat space.
	Let $(y^1,\cdots,y^{p+q})$ be a coordinate system of $\mathbb{R}^{p+q}$ and we define the pseudo Euclidean 
	space $E^{(p,q)}$ equipped with $(p,q)$ type metric 
	\begin{align}
	ds^2=-(dy^1)^2-\cdots-(dy^p)^2+(dy^{p+1})^2+\cdots+(dy^{p+q})^2.
	\end{align}
	In $E^{(2,n-1)}$, $n$ dimensional anti-de Sitter space, AdS$_n$, is given as the submanifold:
	\begin{align}
	-(y^1)^2-(y^2)^2+(y^3)^2+\cdots+(y^{n+1})^2=-1.
	\label{embedded AdS}
	\end{align}
	
	In this paper, we use the global coordinate system of AdS,
	\begin{align}
	y^1&=\sqrt{r^2+1}\cos\tau,\\
	y^2&=\sqrt{r^2+1}\sin\tau,\\
	(y^3)^2&+\cdots+(y^{n+1})^2=r^2. 
	\end{align}
	Then, the induced metric on eq.\eqref{embedded AdS} is given by
	\begin{align}
	ds^2=-(r^2+1)d\tau^2+\frac{1}{r^2+1}dr^2+r^2d\Omega_{n-2}^2,
	\label{AdSglobalmet}
	\end{align}
	where $d\Omega_{n-2}^2$ is the standard metric on $(n-2)$-dimensional sphere.
	
	We consider the conformal diagram to clarify the causal structure of AdS.
	In \eqref{AdSglobalmet} using the coordinate transformation
	\begin{align}
	r=\tan\chi,
	\end{align}
	we obtain
	\begin{align}
	ds^2=\sec^2\chi(-d\tau^2+d\chi^2+\sin^2\chi d\Omega_{n-2}^2).
	\label{AdSconfmetric}
	\end{align}
	Then (the universal cover of) AdS$_n$ is conformal isomorphic to the static Einstein universe 
	(product manifold of sphere and time). 
	Hereafter, we regard AdS$_n$ as the universal cover of eq.\eqref{embedded AdS}. 
	
	For example the conformal diagram of AdS$_3$ is a product manifold of two dimensional half sphere $S^2_{>0}$ and 
	time $\mathbb{R}$. 	(See fig.\ref{AdS_3andLightcorn}.)
	
		\begin{figure}[H]
			\centering
			\includegraphics[keepaspectratio, scale=0.3]{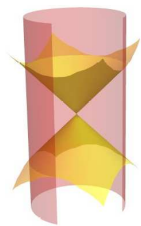}
			\caption{Conformal diagram of AdS$_3$ and light corn.
				AdS$_3$ is conformally embedded in the three dimensional static Einstein universe. 
				AdS$_3$ is the bulk region surrounded by the cylindrical surface.  
				The cylinder is the null infinity.}
			\label{AdS_3andLightcorn}
		\end{figure}
	
	 \section{Classification of Killing vectors of AdS}
	 \label{section Classification of Killing vector field of AdS}
	 AdS$_n$ admits the isometry group SO$(2,n-1)$.
	 Then, there exist $\frac{1}{2}n(n+1)$ Killing vector fields which are linearly independent each other.
	 However, some of them are \lq\lq equivalent" in the following geometrical sense.
	 
	 Let $\mathfrak{so}(2,n-1)$ be the Lie algebra of SO$(2,n-1)$.
	 We define equivalence relation $\sim$ in $\mathfrak{so}(2,n-1)$, i.e., 
	 for $\xi_1, \xi_2\in\mathfrak{so}(2,n-1)$, $\xi_1\sim\xi_2$ if and only if there exists 
	 the $g\in SO(2,n-1)$ such that $\xi_1=Ad(g)\xi_2$. 
	 Let $e^{t\xi}$ be a 1-parameter subgroup of $SO(2,n-1)$ generated by $\xi\in\mathfrak{so}(2,n-1)$.
	 By the action of $e^{t\xi}$ on AdS$_n$. we obtain curves 
	 $e^{t\xi}p,\ (p\in {\rm AdS}_n)$, then 
	 we regard $\xi\in\mathfrak{so}(2,n-1)$ as a Killing vector field tangent to the curves. 
	 Therefore the classification of $\mathfrak{so}(2,n-1)$ induces that of Killing vector fields. 	
	 It is obvious that if $\xi_1\sim\xi_2$, then two base spaces AdS$_n/e^{t\xi_1}$ and AdS$_n/e^{t\xi_2}$ are Riemannian isomorphic. 
	 Therefore, it is sufficient for the complete analysis that we consider base spaces with respect to 
	 the representaive of Killing vectors in this classification. 
	 
	 In the cases of AdS$_3$, AdS$_4$ and AdS$_5$, the Killing vector fields are classified 
	 explicitly\cite{Banados:1992gq, Holst:1997tm, Koike:2008fs, Morisawa:2017lpj}.
	 We regard AdS$_3$ as a submanifold $-t^2-s^2+x^2+y^2=-1$ in the pseudo Euclidean space $E^{(2,2)}$.
	 Then, the Killing vector fields of AdS$_3$ is obtained by restriction of the Killing vector fields of 
	 $E^{(2,2)}$ that generate the rotations and the Lorentz boosts on AdS$_3$. 
	 We can take a set of basis of Lie algebra as 
	 \begin{align}
	 K_{tx}&:=x\partial_t+t\partial_x,\quad
	 K_{ty}:=y\partial_t+t\partial_y,\quad
	 K_{sx}:=x\partial_s+s\partial_x,\quad\\
	 K_{sy}&:=y\partial_s+s\partial_y,\quad
	 L_{ts}:=-s\partial_t+t\partial_s,\quad
	 L_{xy}:=-y\partial_x+x\partial_y.
	 \end{align}
	 According to the notation in ref.\cite{Banados:1992gq}, 
	 the classes of Killing vector fields is shown in the following Table \ref{so(2,2)classifytable}.
	 \begin{table}[H]
	 	\caption{Classification of Killing vector fields in AdS$_3$.}
	 	\begin{center}
	 		\begin{tabular}{|c|c|c|} \hline
	 			Type & Killing vector & \\\hline
	 			$I_a$ & $b(L_{st}-L_{xy})+a(K_{ty}+K_{sx})$ & $ab\ne0,\ a\ne\pm b$  \\
	 			$I_b$& $aK_{tx}+bK_{sy}$ & $a\ne0$\\
	 			$I_c$& $aL_{st}+bL_{xy}$ & $a^2+b^2\ne0$  \\
	 			$II_a$& $a(K_{tx}+K_{sy})+L_{ts}-K_{ty}-K_{sx}-L_{xy}$ & $a\ne 0$ \\
	 			$II_b$& $a(L_{ts}-L_{xy})-L_{ts}-L_{xy}+K_{tx}-K_{sy}$ &   \\
	 			$III^+$& $K_{tx}+L_{xy}$ &\\
	 			$III^-$& $L_{ts}+K_{tx}$ &  \\\hline
	 		\end{tabular}
	 	\end{center}
	 	\label{so(2,2)classifytable}
	 \end{table}	
	 The norm of the Killing vectors are given in the Table \ref{so(2,2)classifynormtable}. 
	 \begin{table}[H]
	 	\caption{Norm of Killing vector fields}
	 	\begin{center}
	 		\begin{tabular}{|c|c|} \hline
	 			Type  & Norm\\\hline
	 			$I_a$ & $a^2-b^2-4ab(tx-sy)$  \\
	 			$I_b$& $-a^2x^2+a^2t^2-b^2y^2+b^2s^2$ \\
	 			$I_c$& $-a^2-(a^2-b^2)(x^2+y^2)$  \\
	 			$II_a$& $-(ax+t-y)^2-(ay-s-x)^2+(at-s-x)^2+(as-t+y)^2$ \\
	 			$II_b$& $1-a^2$ \\
	 			$III^+$& $(t-y)^2$ \\
	 			$III^-$& $-(x-s)^2$ \\\hline
	 		\end{tabular}
	 	\end{center}
	 	\label{so(2,2)classifynormtable}
	 \end{table}
	 
	 We are interested in the possibility of black hole structure, then we omit the type $I_c,II_b,III^-$ 
	 in our study for the following reasons. 
	 In type $III^-$ norm-weighted spacetime has Euclidean metric since the norm of Killing vector is negative definite.
	 In type $II_b$ the norm of the Killing vector is constant so its base space dose not have singularity.
	 In type $I_c$, if $a^2>b^2$ then the norm is negative definite and if $a^2<b^2$ then there exists naked singularity. 
	 Therefore we clarify the causal structure of base spacetimes by Killing vectors of type 
	 $I_a,I_b,II_a,III^+$ of AdS$_3$.
	 
	 \section{Method of analysis for causality}\label{section Analysis Method of Causality of norm-weighted Orbit Spacetime}
	 
	 In this section we explain the method to clarify the causal structure of base space 
	 used in this paper.
	 Let $(\tilde{M},\tilde{g})$ be a Lorentzian manifold, and $\xi$ be a Killing vector field. 
	 And let $(M=\tilde{M}/\xi,\hat{g})$ be the base space with respect to $\xi$ and we assume that $|\xi|^2$ has zero points.
	 
	 In $(M,\hat{g})$, the region where the norm of $\xi$ vanishes, i.e., $|\xi|^2=0$, is curvature singularity 
	 since the conformal factor in eq.\eqref{decomposed_fields} vanishes. 
	 Even though we can not take an Einstein frame by eq.\eqref{einstein frame alpha} in the case of $n=2$, 
	 the theory can not be defined on the region $|\xi|^2=0$ in eq.\eqref{reduction action} since the integration 
	 measure vanishes there.
	 Therefore we also regard $|\xi|^2=0$ as a singularity in the $n=2$ case.
	 
	 We consider a region in $\tilde{M}$ where $\xi$ is spacelike so that $(M,\hat{g})$ becomes Lorentzian manifold.
	 We construct conformal diagrams so that we clarify the causal structure of base space. 
	 For this aim, we introduce the null polar coordinate of the base space.
	 
	 We construct the null polar coordinate in the case of AdS$_n$.
	 The global coordinate of AdS$_n$ is given by \eqref{AdSconfmetric}. 
	 In this coordinate a null geodesic which passes through a point $N(\tau_0)$ specified $(\tau,\chi)=(\tau_0,0)$ 
	 and the direction of the geodesic tangent is specified by $\Omega_0$ is given by
	 \begin{align}
	 \tau-\tau_0=\lambda,\quad \chi=\lambda,\quad \Omega=\Omega_0,\label{eq of null geodesic in AdS3}
	 \end{align}
	 where $\lambda$ is a parameter along the null geodesic.
	 
	 Here we note the horizontal lift used in a context of fiber bundle.
	 Let $\pi:\tilde{M}\rightarrow M=\tilde{M}/\xi$ be a Riemannian submersion, and $c$ be a curve passing through 
	 a point $p\in M$. 
	 For a point $q\in\pi^{-1}(p)$, there exists a curve $\tilde{c}$ passing through $q$ such that $\pi(\tilde{c})=c$.
	 This curve $\tilde{c}$ is called a lift of the curve $c$ passing through $q$. 
	 And the curve $\tilde{c}$ is called the horizontal lift of the curve $c$ passing through $q$ 
	 if $\tilde{c}$ lies in horizontal subspace. 
	 The horizontal lift $\tilde{c}$ is unique for $c$ and $q$. 
	 
	 Let $\pi:AdS_n\rightarrow AdS_n/\xi$ be a Riemannian submersion with respect to the Killing vector $\xi$.
	 The horizontal lift of a null geodesic which passes through a point $\pi(N(\tau_0))$ in AdS$_n/\xi$ is 
	 a null geodesic which passed through the point $N(\tau_0)$ and is orthogonal to $\xi$. 
	 We call it horizontal null geodesic whose direction is $\Omega^\parallel_0$.
	 The point on the horizontal null geodesic is labeled by a set of parameters $(\tau_0,\lambda,\Omega^\parallel_0)$.
	 
	 We consider the region of AdS$_n/\xi$ such that any point is connected to the timelike curve $\pi(N(\tau))$ 
	 by a null geodesic.
	 By the projection map $\pi$ we can regard the set $(\tau,\lambda,\Omega^\parallel_0)$ as the coordinate of 
	 the corresponding region of AdS$_n/\xi$. 
	 By using this coordinate, we specify the position of the singularity in AdS$_n/\xi$. 
	 For this purpose we investigate the norm of $\xi$ on horizontal null geodesics $c(\lambda)$ that pass 
	 through $N(\tau_0)$ of AdS$_n$ as 
	 \begin{align}
	 f(\lambda):=|\xi(c(\lambda))|^2 . 
	 \label{definition of f}
	 \end{align}
	 If $f(\lambda)$ has a zero point at a finite $\lambda_0$, the curve $c(\lambda)$ intersects with the surface 
	 $|\xi|^2=0$ there, namely, the curve $\pi(c(\lambda_0))$ hits the singularity. 
	 
	 \subsection{Intersections of null geodesics and singularity}
	 \label{subsection The intersection of null geodesics and singularity}
	 In this subsection we will give the necessary and sufficient condition for a null geodesic intersects 
	 with the singularity, $|\xi|^2=0$. 
	 The following proposition concerning to geodesics is known for general constant curvature space.
	 \begin{prop}
	 	A geodesic passes through a point $y_0$ of the submanifold
	 	\begin{align}
	 	M:\ \eta_{\mu\nu}y^\mu y^\nu=\varepsilon a^2\ (\varepsilon=\pm1),
	 	\end{align}
	 	in the pseudo Euclidean space $E^{(p,q)}$
	 	is given by an intersection curve of M and two dimensional plane in $E^{(p,q)}$ that contains 
	 	the normal vector on the point $y_0$.
	 	Especially, if M is Lorentzian manifold a null geodesic is a null line in $E^{(p,q)}$.
	 	\label{geodesic of AdS}
	 \end{prop}
	 In the following discussion we often identify $E^{(2,n-1)}$ and its tangent space $T_pE^{(2,n-1)}$ 
	 as a vector space $E^{(2,n-1)}$.
	 We regard AdS$_n$ is embedded in $E^{(2,n-1)}$, then 
	 an arbitrary point $N$ of $AdS$ is a point of $E^{(2,n-1)}$ then $N$ is timelike unit vector, $|N|^2=-1$.
	 Furthermore, $N$ is regarded as a normal vector of $AdS$ on the point $N$, 
	 a vector which is orthogonal to the vector $N$ is regarded as an element of $T_N(AdS)$ naturally. 
	 Let $X$ be a null vector which is orthogonal to the vector $N$ then $X$ is a null vector on AdS$_n$. 
	 By proposition\ref{geodesic of AdS}, a null geodesic which pass through the point $N$ of AdS$_n$ and 
	 has tangent vector $X$ is given by
	 \begin{align}
	 x(\lambda)=N+\lambda X.
	 \label{null equation}
	 \end{align}
	 
	 Moreover, let $\rho$ be a representation $\rho:SO(2,n-1)\rightarrow GL(E^{(2,n-1)})$, 
	 then a representation matrix of $\xi\in \mathfrak{so}(2,n-1)$ is $\rho_\ast(\xi)$. 
	 We write the same symbol $\xi$ for $\rho_\ast(\xi)$ for brevity. 
	 The matrix $\xi$ generates the vector field $\xi^\ast$ as 
	 \begin{align}
	 \xi^\ast(P):=\xi P,
	 \end{align}
	 where P is a point in $E^{(2,n-1)}$. 
	 Since $e^{t\xi}$ is a 1-parameter subgroup of the isometry $SO(2,n-1)$, 
	 for an arbitrary vectors $Z,W\in E^{(2,n-1)}$ we have
	 \begin{align}
	 \eta(e^{t\xi}Z,e^{t\xi}W)=\eta(Z,W),
	 \end{align}
	 where $\eta$ is the metric of $E^{(2,n-1)}$. 
	 Differentiating this equation with $t$ and putting $t=0$, we obtain
	 \begin{align}
	 \eta(\xi Z,W)+\eta(Z,\xi W)=0.
	 \label{linear killing equation}
	 \end{align}	
	 On a point $N\in AdS_n$, the orthogonal condition of a tangent vector $Z$ and 
	 the Killing vector $\xi^\ast$ is given by 
	 \begin{align}
	 \eta(X,\xi^\ast)_{N}=\eta(X,\xi N)=\eta(\xi X, N)=0,
	 \end{align}
	 where we regard $X,\xi N$ and $\xi X$ as vectors of $E^{(2,n-2)}$.
	 
	 The next proposition and its corollary are useful in this paper.
	 \begin{prop}
	 	Let N be an arbitrary point of $AdS_n(\subset E^{(2,n-1)})$.
	 	Let $\xi^\ast$ be a Killing vector field and X be a null vector at N in AdS$_n$ such that 
	 	$|X|^2=0,\ \eta(X,\xi N)=0$ and is not eigenvector
	 	of $\xi$.
	 	For a null geodesic $x(\lambda)$ which pass through $N$ and has tangent vector $X$, 
	 	the followings are equivalent.
	 	
	 	(1)\qquad
	 	$x(\lambda)$ intersects with hypersurface $|\xi^\ast|^2=0$.
	 	
	 	(2)\qquad
	 	$
	 	H(\xi N,\xi X):=|\xi N|^2|\xi X|^2-\eta(\xi N,\xi X)^2=0.
	 	$
	 	
	 	(3)\qquad
	 	$X$, $\xi X$, $N$ are not linearly independent.  
	 	\label{main prop}
	 \end{prop}
	 \begin{proof}
	 	$(1)\Leftrightarrow(2)$
	 	
	 	From eq.\eqref{null equation} we have
	 	\begin{align}
	 	\xi^\ast(x(\lambda))=\xi x(\lambda)=\xi N+\lambda\xi X,
	 	\label{xi_ast}
	 	\end{align}
	 	and it is assumed
	 	\begin{align}
	 	\eta(N,X)=\eta(X,X)=\eta(X,\xi N)=0
	 	\label{main prop assumption}
	 	\end{align}
	 	on N then
	 	\begin{align}
	 	\eta(\xi^\ast(x(\lambda)),X)&=\eta(\xi N,X)+\lambda\eta(\xi X,X)=0,\\
	 	\eta(\xi^\ast(x(\lambda)),N)&=\eta(\xi N,N)+\lambda\eta(\xi X,N)=0,
	 	\end{align}
	 	where we have used eq.\eqref{linear killing equation}.
	 	Therefore
	 	\begin{align}
	 	\xi^\ast(x(\lambda))\in X^\perp\cap N^\perp, 
	 	\end{align}
	 	where $X^\perp$ denotes the orthogonal complement subspace to $X$, and $N^\perp$ denotes 
	 	the same for $N$. 
	 	It means $\xi^\ast(x(\lambda))$ is not timelike hence
	 	\begin{align}
	 	|\xi^\ast(x(\lambda))|^2\geq0.
	 	\label{main prop nontimelike cond}
	 	\end{align}
	 	
	 	And using eq.\eqref{xi_ast} the norm of $\xi^\ast$ on $x(\lambda)$ is given as quadratic of $\lambda$:
	 	\begin{align}
	 	f(\lambda):=|\xi^\ast(x(\lambda))|^2=|\xi X|^2\lambda^2+2\eta(\xi N,\xi X)\lambda+|\xi N|^2,
	 	\label{flambda}
	 	\end{align}
	 	then $f(\lambda)$ has zeros if and only if
	 	\begin{align}
	 	\eta(\xi N,\xi X)^2-|\xi X|^2|\xi N|^2\geq0.
	 	\end{align}
	 	
	 	Together with eq.\eqref{main prop nontimelike cond} we have
	 	\begin{align}
	 	|\xi X|^2|\xi N|^2-\eta(\xi N,\xi X)^2=0.
	 	\label{conditionH}
	 	\end{align}

	 	$(2)\Leftrightarrow(3)$
	 	
	 	If $\xi X,~ \xi N$ are linearly dependent then \eqref{conditionH} holds.
	 	We suppose that \eqref{conditionH} holds for linearly independent vectors $\xi X,\xi N$, 
	 	then the subspace $<\xi N,\xi X>_{{\rm span}}$ is degenerate with respect to the induced metric 
	 	from $\eta$. 
	 	By eqs. \eqref{linear killing equation} and \eqref{main prop assumption} 
	 	then $\xi X\in X^\perp\cap N^\perp$ holds.
	 	Similarly $\xi N\in X^\perp\cap N^\perp$ holds.
	 	Since $X,\xi N,\xi X\in X^\perp\cap N^\perp$ and $<\xi N,\xi X>_{\rm span}$ is degenerate therefore $X\in<\xi N,\xi X>_{\rm span}$.
	 	Therefore, in the both cases $X,\xi X,\xi N$ are not lineary independent.
	 	The inverse is easily verified. 
	 \end{proof}
	 
	 \begin{cor}\label{cor intersection point}
	 	If $X$ is not eigenvector of $\xi$ and $X,\xi X,\xi N$ are not lineary independent, 
	 	then the null curve $x(\lambda)$ intersects with the hypersurface $|\xi^\ast|^2=0$ at the point 
	 	\begin{align}
	 	\lambda_0=-\frac{\eta(\xi N,\xi X)}{|\xi X|^2}.
	 	\label{lambda0 eq}
	 	\end{align}
	 	We define the function $\phi:=|\xi^\ast|^2$ on $E^{(2, n-1)}$ then the sign of $\lambda_0$ is 
	 	the same as that of $-d\phi(X)_N$.
	 \end{cor}
	 \begin{proof}
	 	Equation \eqref{lambda0 eq} is obvious from \eqref{flambda}.
	 	Let $\{x^\mu\}$ be a orthonormal coordinate system of $E^{(2,n-1)}$.
	 	We have 
	 	\begin{align}
	 	\xi X=\xi^\mu_\nu X^\nu \partial_\mu=X^\mu\partial_\mu(\xi^\rho_\nu x^\nu)\partial_\rho
	 	=X^\mu\partial_\mu\xi^{\ast\rho}\partial_\rho=\nabla_X\xi^\ast,
	 	\end{align}
	 	where $\nabla$ is Riemannian connection with respect to $\eta$, then we obtain
	 	\begin{align}
	 	\eta(\xi X,\xi N)=\eta(\nabla_X\xi^\ast,\xi^\ast)_N=\frac{1}{2}\nabla_X|\xi^\ast|_N^2=\frac{1}{2}d\phi(X)_N. 
	 	\end{align}
	 	Since $\xi X$ is spacelike, $|\xi X|^2>0$, then the sign of $\lambda_0$ coincides with that of $-d\phi(X)_N$.
	 \end{proof}
	 Especially in the case of $n=3$, i.e., the case of AdS$_3$, the following fact holds.
	 
	 \begin{cor}\label{AdS3degenerate}
	 	Let $N$ be an arbitrary point of $AdS_3(\subset E^{(2,2)})$.
	 	Let $\xi^\ast$ be a Killing vector field, and $X\in E^{(2,2)}$be a null vector at a point N 
	 	in AdS$_3$.
	 	And assume $\eta(X,\xi N)=0$ and $|\xi^\ast|^2$ has zero points.
	 	Then a null geodesic which pass through an arbitrary N and have tangent vector X which is not eigen vector of $\xi$ has intersection with singularity $|\xi^\ast|^2=0$.
	 \end{cor}
	 \begin{proof}
	 	Since $\eta$ is non degenerate metric of $E^{(2,2)}$, $\dim (X^\perp\cap N^\perp)=2$ holds.
	 	And we see $X,\xi X,\xi N\in X^\perp\cap N^\perp$ easily, then $X,\xi X,\xi N$ are linearly dependent.
	 	From Proposition.\ref{main prop} the statement holds.
	 \end{proof}
	 
	 We also consider the case that $X$ is an eigenvector of $\xi$.
	 
	 \begin{prop}\label{constantxi}
	 	Let N be an arbitrary point of $AdS_n(\subset E^{(2,n-1)})$.
	 	Let $\xi^\ast, X\in E^{(2,n-1)}$ be a Killing vector field and a null vector on N in AdS$_n$ respectively.
	 	And let $X$ be an eigenvector of $\xi$ satisfying $\eta(X,\xi N)=0$.
	 	For a null geodesic $x(\lambda)$ which pass through $N$ and has tangent vector $X$,
	 	\begin{align}
	 	f(\lambda)=|\xi^\ast(x(\lambda))|^2=|\xi(x(0))|^2=Const.
	 	\end{align}
	 	holds.
	 \end{prop}
	 \begin{proof}
	 	We assume $\xi X=\alpha X$, where $\alpha$ is a constant. 
	 	Then $|\xi X|^2=\alpha^2|X|^2=0$.
	 	Therefore, 
	 	\begin{align}
	 	|\xi^\ast(x(\lambda))|^2&=|\xi X|^2\lambda^2+2\eta(\xi N,\xi X)\lambda+|\xi N|^2\\
	 	&=2\eta(\xi N,\xi X)\lambda+|\xi N|^2.
	 	\end{align}
	 	From Proposition \ref{main prop} we have
	 	\begin{align}
	 	|\xi^\ast(x(\lambda))|^2\geq0
	 	\end{align}
	 	then
	 	$\eta(\xi N,\xi X)=0$
	 	and 
	 	\begin{align}
	 	|\xi^\ast(x(\lambda))|^2=|\xi N|^2.
	 	\end{align}
	 \end{proof}
	 
	 Finally we show that the singularity surface in base space is not spacelike.
	 Let $\pi:(\tilde{M},\tilde{g})\rightarrow (M,g)$ be a Riemannian submersion with respective to a Killing vector $\xi$, where $g$ is projective metric by $\xi$.
	 We assume that the function $|\xi|^2:\tilde{M}\rightarrow \mathbb{R}$ dose not have a critical point.
	 Then for an arbitrary point $p$ and its neighburhood $U$ the set $\{x\in\tilde{M};|\xi(x)|^2=|\xi(p)|^2\}$ defines 
	 a hypersurface in $U$ by implicit function theorem. 
	 
	 Let ${\rm S}_\varepsilon$ be the hypersurface $|\xi|^2=\varepsilon$
	 on $\tilde{M}$. 
	 And we assume that there exists a Killing vector $\eta$ which commutes with $\xi$ and is linearly independent. 
	 Indeed, we can take such a Killing vector $\eta$ in our case. 	
	 Since the Killing vector $\eta$ that commutes with $\xi$ is tangents to ${\rm S}_\varepsilon$, 
	 then $\eta':=\pi_\ast\eta$ also tangents to $\pi({\rm S}_\varepsilon)$.
	 Then the norm of $\eta'$ on $\pi({\rm S}_\varepsilon)$ is given by
	 \begin{align}
	 g(\eta',\eta')&=\tilde{g}(\eta,\eta)-\frac{1}{|\xi|^2}\tilde{g}(\xi,\eta)^2\\
	 &=\tilde{g}(\eta,\eta)-\frac{1}{\varepsilon}\tilde{g}(\xi,\eta)^2.
	 \end{align}
	 
	 Now we consider in the region $|\xi|^2\geq0$ in $\tilde{M}$.
	 If $\eta$ is not orthogonal to $\xi$ on ${\rm S}_0$, since $\lim\limits_{\varepsilon\rightarrow+0}g(\eta',\eta')=-\infty$, the Killing vector $\eta'$ is timelike on the singularity surface $\pi(S_\varepsilon)$ in M for a small $\varepsilon$.
	 And if $\eta$ is orthogonal to $\xi$ on ${\rm S}_0$ then $\eta$ is null or spacelike, however $\xi$ and $\eta$ are linearly independent therefore $\eta$ is spacelike on ${\rm S}_0$.
	 Since the subspace which is spaned by $\xi$ and $\eta$ in tangent space of ${\rm S}_0$ degenerate, so we have
	 \begin{align}
	 \tilde{g}(\xi,\xi)\tilde{g}(\eta,\eta)-\tilde{g}(\xi,\eta)^2=0.
	 \end{align}
	 It implies $\lim\limits_{\varepsilon\rightarrow+0}g(\eta',\eta')=0$.
	 
	 In the case of $\tilde{M}={\rm AdS}_3$, since there exists a Killing vector $\eta$ such above, the singularities of base space by $\xi$ are timelike or null curves.
	 
	 \begin{prop}\label{prop killing vector on singularity cureve}
	 	Let $\xi$ be a Killing vector field in AdS$_3$.
	 	We assume that $\xi$ dose not have a fixed point.
	 	Then the singularity curve $|\xi|^2=0$ in norm-wighted orbit space by $\xi$ is timelike or null curve.
	 \end{prop}

	 \section{Causality of the base space in AdS$_3$}
	 \label{section Causality in norm-weighted Orbit Space of AdS_3}
	 In this section we investigate the causal structure of base spaces of AdS$_3$ which are specified by the Killing vectors used for the submersion. 
	 We concentrate on the Killing vectors in the types $I_a,I_b,II_a,III^+$ listed in Table\ref{so(2,2)classifytable} 
	 because these type would lead to non-trivial causal structure of base space.

	 In AdS$_3$, from corollary \ref{AdS3degenerate}, almost all null geodesics which are orthogonal to $\xi$ 
	 have intersection with singularity $|\xi|^2=0$.
	 The position of the intersection is given by the use of corollary \ref{cor intersection point}.
	 
	 Let $\{t,s,x,y\}$ and $\eta$ be a orthonormal coordinate system  and the metric of $E^{(2,2)}$, respectively.
	 We use the global coordinate \eqref{AdSglobalmet} for AdS$_3$ in the form 
	 \begin{align}
	 t&=\sqrt{r^2+1}\sin\tau,\\
	 s&=\sqrt{r^2+1}\cos\tau,\\
	 x&=r\cos\theta,\\
	 y&=r\sin\theta.
	 \end{align}
	 
	 Our purpose is the analysis of a null geodesic passing through each point $N(\tau)$ on the curve $r=0$, 
	 \begin{align}
	 N(\tau):={}^T(\sin\tau,\cos\tau,0,0).
	 \end{align}
	 A future directed tangent vector of the null geodesic at $N$ is given by
	 \begin{align}
	 X={}^T(\cos\tau,-\sin\tau,\cos\theta_0,\sin\theta_0)
	 \end{align}
	 where $\theta_0$ is a parameter that specify the direction of $X$. 
	 
	 First, we show that the type $I_b,II_a$ is black hole, and secondly, 
	 the type $I_a,III^+$ have naked singularity and null singularity.
	 
	 \subsection{Type$I_b:\xi=K_{tx}+a K_{sy}$}
	 The representation matrix of $\xi$ is	
	 \begin{align}
	 \xi&=\left(
	 \begin{matrix}
	 0&0&1&0\\
	 0&0&0&a\\
	 1&0&0&0\\
	 0&a&0&0
	 \end{matrix}\right),
	 \end{align}
	 then we have
	 \begin{align}
	 \phi=\frac{1}{2}|\xi^\ast|^2=\frac{1}{2}(-x^2-a^2 y^2+t^2+a^2 s^2).
	 \end{align}
	 If $a=1$, we have $|\xi|^2=2\phi=1$ then there exists no singularity. 
	 If $0<a<1$, 
	 we have 
	 \begin{align}
	 d\phi=-xdx-a^2 ydy+tdt+a^2 sds,
	 \end{align}
	 then
	 \begin{align}
	 d\phi(X)_N=\sin\tau\cos\tau-a^2\sin\tau\cos\tau=\frac{1}{2}(1-a^2)\sin2\tau.
	 \label{I_b dphiX}
	 \end{align}
	 
	 Since $d\phi(X)_N$ is independent on the direction parameter $\theta_0$, eq.\eqref{I_b dphiX} holds for the null vector $X$ which is orthogonal to $\xi^\ast$.
	 As shown in the corollary \ref{cor intersection point}, the sign of $\lambda_0$ is the same that of $-d\phi(N)$, 
	 that is, ${\rm sign}(\lambda_0)= {\rm sign}(-d\phi(N))$. The function $-d\phi(N)$ is shown as a function of $\tau$ 
	 in the figure \ref{typeIblambda0sign}. 
	 
	 \begin{figure}[H]
	 	\centering
	 	\includegraphics[keepaspectratio, scale=0.3]{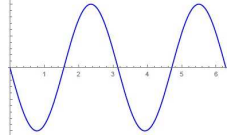}
	 	\caption{Vertical axis is $-d\phi(X)_N$ and horizontal axis is $\tau$.}
	 	\label{typeIblambda0sign}
	 \end{figure}
	 
	 In the region of $\tau$ such that $-d\phi(N)>0$, the null geodesics 
	 intersect with the singularity in the future and are extendible to the past infinity, while in the region of $\tau$ 
	 such that $-d\phi(N)<0$, that intersect with the singularity in the past and are extendible to the future infinity.
	 
	 In the case of $\tau=n\pi/2$, that is $d\phi(X)=\eta(\xi X,\xi N)=0$, we need to clarify the behavior of endpoints 
	 of the null geodesics, which depend on the norm of $\xi^\ast$ at the point N.
	 If $|\xi N|^2\ne0$, then null geodesics dose not intersect with the singularity 
	 since $f(\lambda)=|\xi X|^2\lambda^2+|\xi N|^2\geq|\xi N|^2=f(0)>0$.
	 Therefore the null geodesics are extendible to the future and past infinity.
	 If $|\xi N|^2=0$, the null geodesics intersect with the singularity at $N$ i.e. $\lambda_0=0$.
	 
	 In the case of type $I_b$ we have
	 \begin{align}
	 |\xi N|^2=\phi(N)=\sin^2\tau+a^2\cos^2\tau,
	 \end{align}
	 then we should consider the two cases $a=0$, $|\xi N|$ can be 0, and $a\ne0$, $|\xi N|$ never vanishes, 
	 separately.
	 
	 In the $a=0$ case, 
	 we have $|\xi N|^2=\sin^2\tau$. 
	 For $\tau=(2n+1)\pi/2\ n\in\mathbb{Z}$, we have $f(0)=1\ne0$ then the null geodesics are extendible to 
	 the both future and past infinity. 
	 On the other hand, for $\tau=n\pi\ n\in\mathbb{Z}$, $f(0)=0$ holds, 
	 then the spacetime is bounded in the domain $0<\tau<\pi$. 
	 In $0<\tau<\pi/2$ null geodesics pass through N intersect with the singularity in the past and are extendible to the future null infinity.
	 
	 In $\pi/2<\tau<\pi$ null geodesics passing through N intersect with the singularity in the future 
	 and are extendible to the past null infinity. 
	 Null infinity of AdS$_3/\xi$ is timelike as that of AdS$_3$, hence the causal structure of this case is shown by the conformal diagram \ref{penrosek_tx}.
	 
	 \begin{figure}[H]
	 	\centering
	 	\includegraphics[keepaspectratio, scale=0.4]{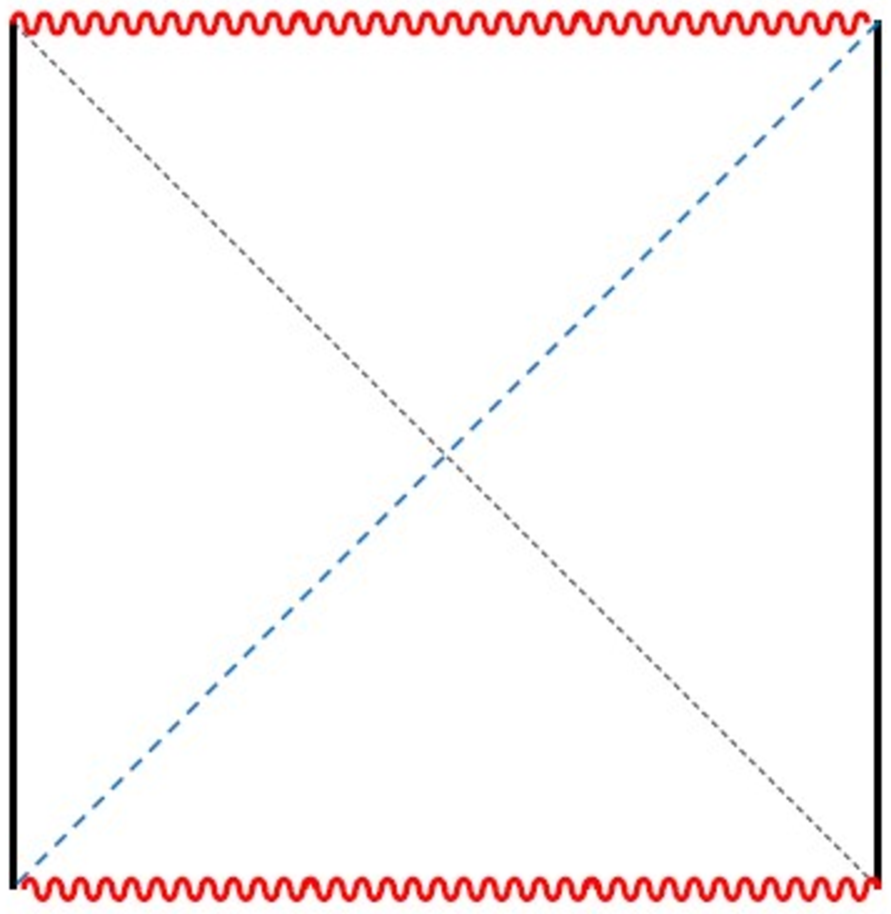}
	 	\caption{Conformal diagram of the case of $a=0$.
	 		The wavy line (red) is the singularity and the dotted line (blue) is the event horizon and the vertical solid line is the null infinity.}
	 	\label{penrosek_tx}
	 \end{figure}

	 In the $a\ne0$ case,	
	 $f(0)=a^2$ holds for $\tau=n\pi,n\in\mathbb{Z}$, and $f(0)=1$ holds for $\tau=(n+\frac{1}{2})\pi,n\in\mathbb{Z}$.
	 The conformal diagram is shown in Fig. \ref{penroseK_tx+aK_sy}. 
	 
	 \begin{figure}[H]
	 	\centering
	 	\includegraphics[keepaspectratio, scale=0.6]{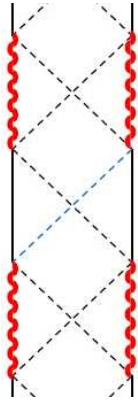}
	 	\caption{Conformal diagram of the case of $a\ne0$.}
	 	\label{penroseK_tx+aK_sy}
	 \end{figure}
	 
	 \subsection{Type$II_a:\ \xi=a(K_{tx}+K_{sy})+L_{ts}-K_{ty}-K_{sx}-L_{xy}$}
	 The representation matrix is
	 \begin{align}
	 \xi=\left(
	 \begin{matrix}
	 0&-1&a&-1\\
	 1&0&-1&a\\
	 a&-1&0&1\\
	 -1&a&-1&0
	 \end{matrix}\right),
	 \end{align}
	 then we have
	 \begin{align}
	 \phi=|\xi|^2=-(-s+ax-y)^2-(t-x+ay)^2+(-s+at+y)^2+(as-t-x)^2.
	 \end{align}
	 And we obtain 
	 \begin{align}
	 d\phi_N=2a(at-2s)dt+2a(as-2t)ds+4atdx-4asdy
	 \end{align}
	 and
	 \begin{align}
	 d\phi(X)|_N=4(\sin(\tau-\theta_0)-a\cos2\tau).
	 \end{align}
	 
	 Since we consider the case that $\xi$ is spacelike, then we consider the region
	 \begin{align}
	 \phi(N)=a(a-2\sin2\tau)>0,\\
	 \Leftrightarrow
	 \frac{a}{2}\geq|\sin2\tau|.
	 \end{align}
	 
	 If $0<a<2$, we take
	 \begin{align}
	 \tau_i&\le\tau\le\tau_f,\\
	 \tau_i&=\pi/2-\arcsin(a/2)/2,\\
	 \tau_f&=\pi+\arcsin(a/2)/2,
	 \end{align}
	 as such a domain and if $a>2$ all $\tau$ are allowed.
	 
	 Orthogonal condition of $\xi^\ast_N(=\xi N)$ and $X$ is given by
	 \begin{align}
	 \eta(\xi^\ast,X)_N=1 + \cos\theta_0(a \sin\tau - \cos\tau) + 
	 \sin\theta_0 (-\sin\tau + a \cos\tau) = 0.
	 \label{IIaorthogonal}
	 \end{align}
	 We draw $-d\phi(X^\pm)|_N$ numerically for $\theta_0=\theta_0^\pm(\tau)$, where $\theta_0=\theta^+_0(\tau),\theta^-_0(\tau)$ are the solutions of equation\eqref{IIaorthogonal}.	
	 We obtain $-d\phi(X^\pm)|_N$ as a function of $\tau$ as shown in the Figure \ref{typeIIa lambda0 sign 0<a<1},\ref{typeIIa lambda0 sign a=1},\ref{typeIIa lambda0 sign a=sqrt2},\ref{typeIIa lambda0 sign 1<a<2},\ref{typeIIa lambda0 sign a>2}.
	 
	 \begin{figure}[H]
	 	\begin{minipage}{0.3\hsize}
	 		\includegraphics[scale=0.3]{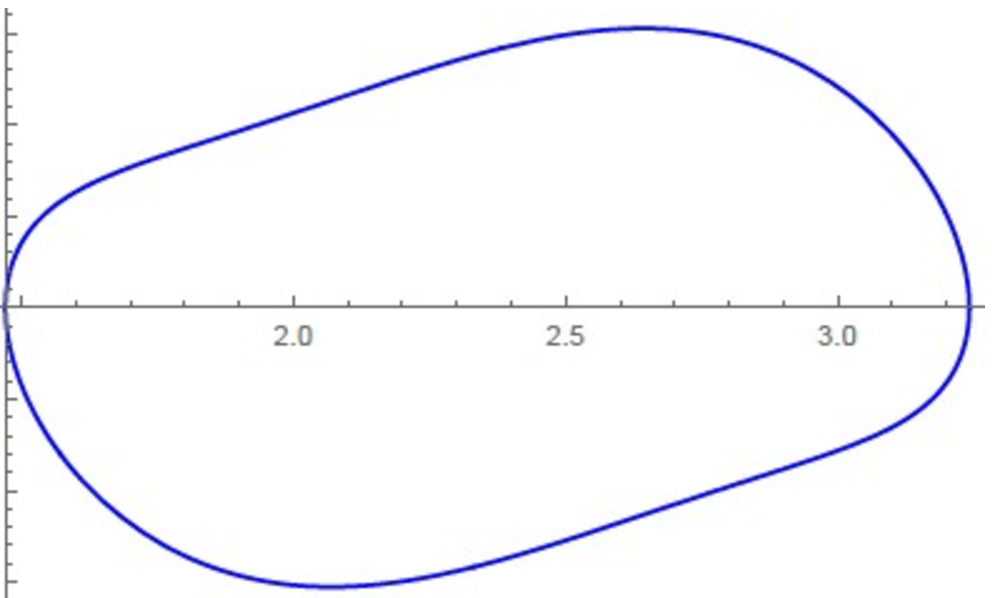}
	 		\caption{Case of $0<a<1$.
	 			Horizontal axis is $\tau$ and vertical axis is $-d\phi(X)|_N$.}
	 		\label{typeIIa lambda0 sign 0<a<1}
	 	\end{minipage}
	 	\begin{minipage}{0.3\hsize}
	 		\includegraphics[scale=0.3]{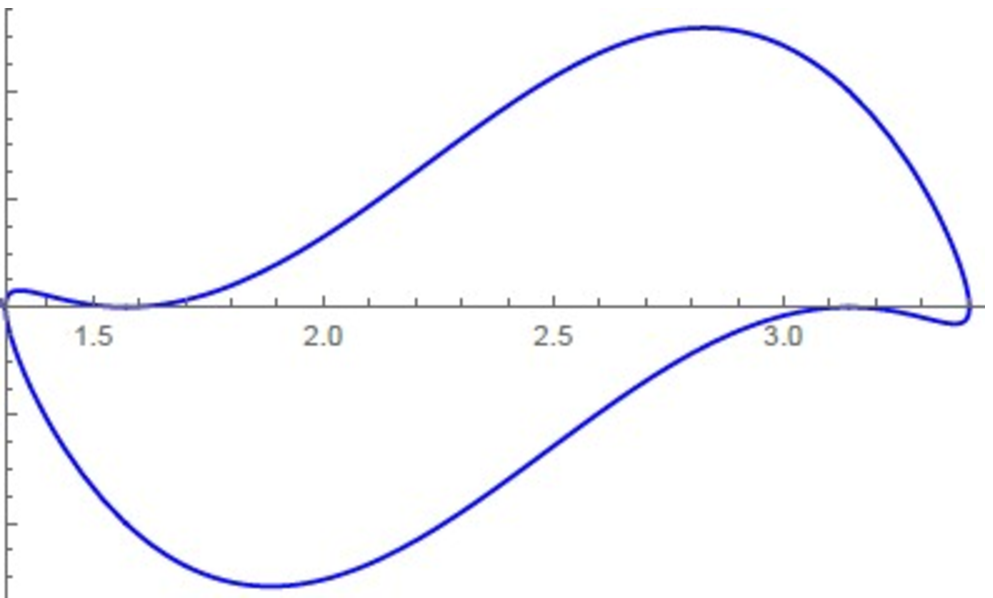}
	 		\caption{Case of $a=1$.}
	 		\label{typeIIa lambda0 sign a=1}
	 	\end{minipage}
	 	\begin{minipage}{0.3\hsize}
	 		\includegraphics[scale=0.3]{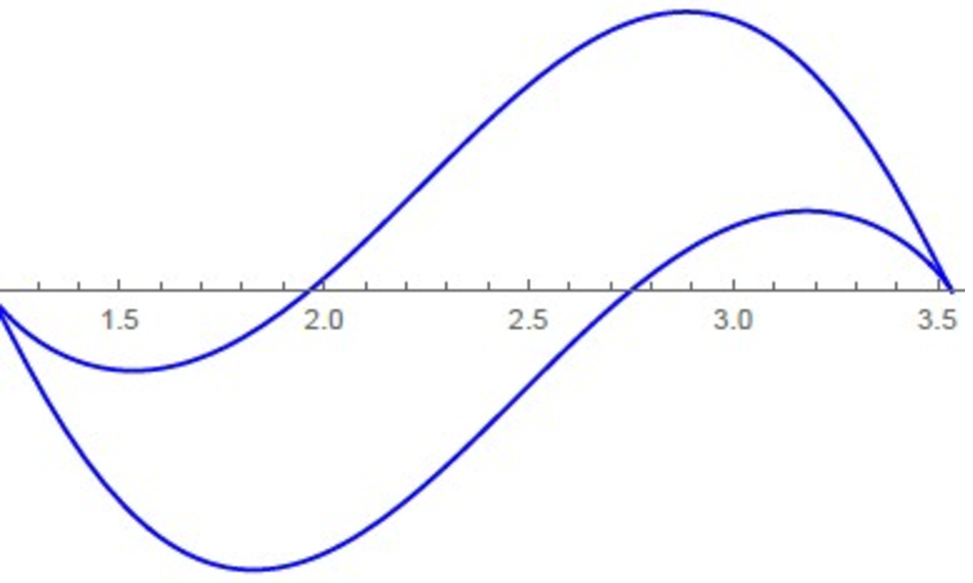}
	 		\caption{Case of $a=\sqrt{2}$.}
	 		\label{typeIIa lambda0 sign a=sqrt2}
	 	\end{minipage}
	 \end{figure}
	 \begin{figure}[H]
	 	\begin{minipage}{0.5\hsize}
	 		\includegraphics[scale=0.5]{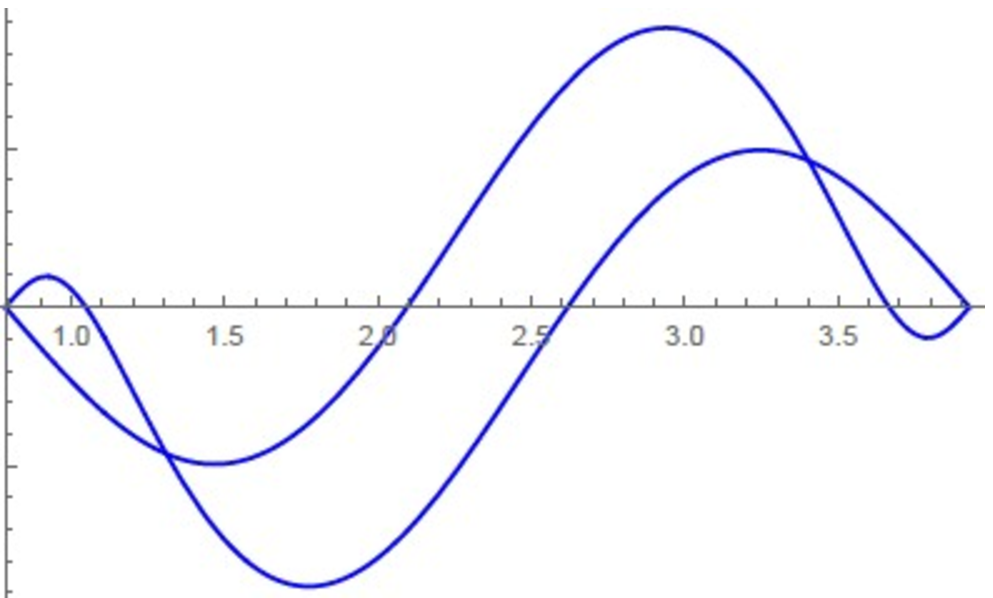}
	 		\caption{Case of $1<a<2\ (a\ne\sqrt{2})$.}
	 		\label{typeIIa lambda0 sign 1<a<2}
	 	\end{minipage}
	 	\begin{minipage}{0.5\hsize}
	 		\includegraphics[scale=0.5]{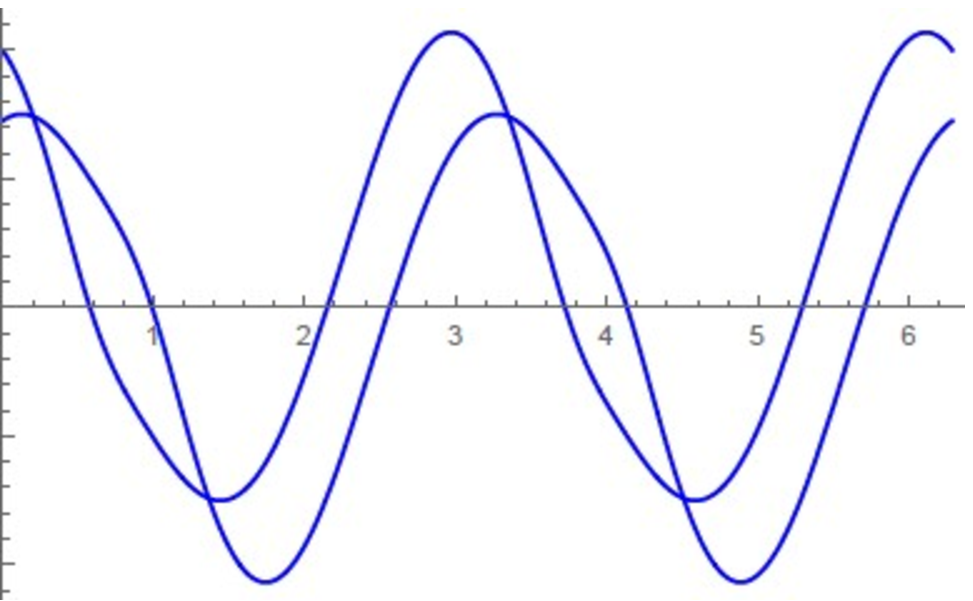}
	 		\caption{Case of $a>2$.}
	 		\label{typeIIa lambda0 sign a>2}
	 	\end{minipage}
	 \end{figure}
	 
	 Then we obtain the conformal diagrams as shown in the following Figures \ref{penroseTypeIIaafrom0to1},\ref{penroseTypeIIaa1},\ref{penroseTypeIIaasqrt2},\ref{penroseTypeIIaafrom1to2},\ref{penroseTypeIIaafrom2to}.
	 
	 \begin{figure}[H]
	 	\begin{minipage}{0.5\hsize}
	 		\centering\includegraphics[scale=0.3]{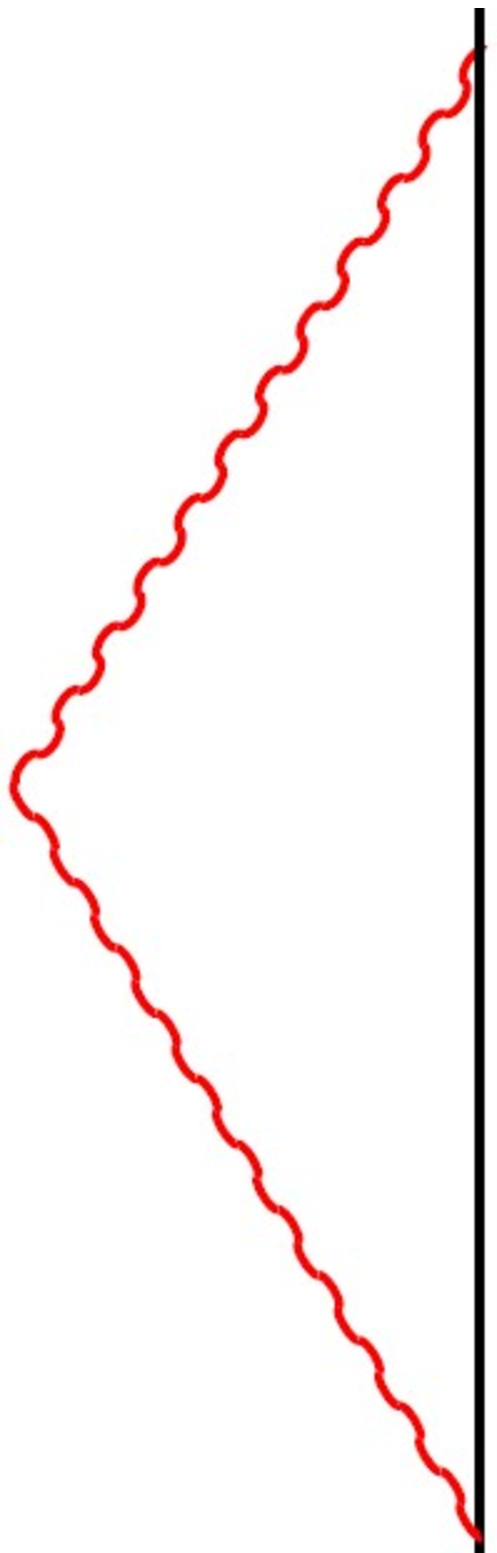}
	 		\caption{$0<a<1$}
	 		\label{penroseTypeIIaafrom0to1}
	 	\end{minipage}
	 	\begin{minipage}{0.5\hsize}
	 		\centering\includegraphics[scale=0.3]{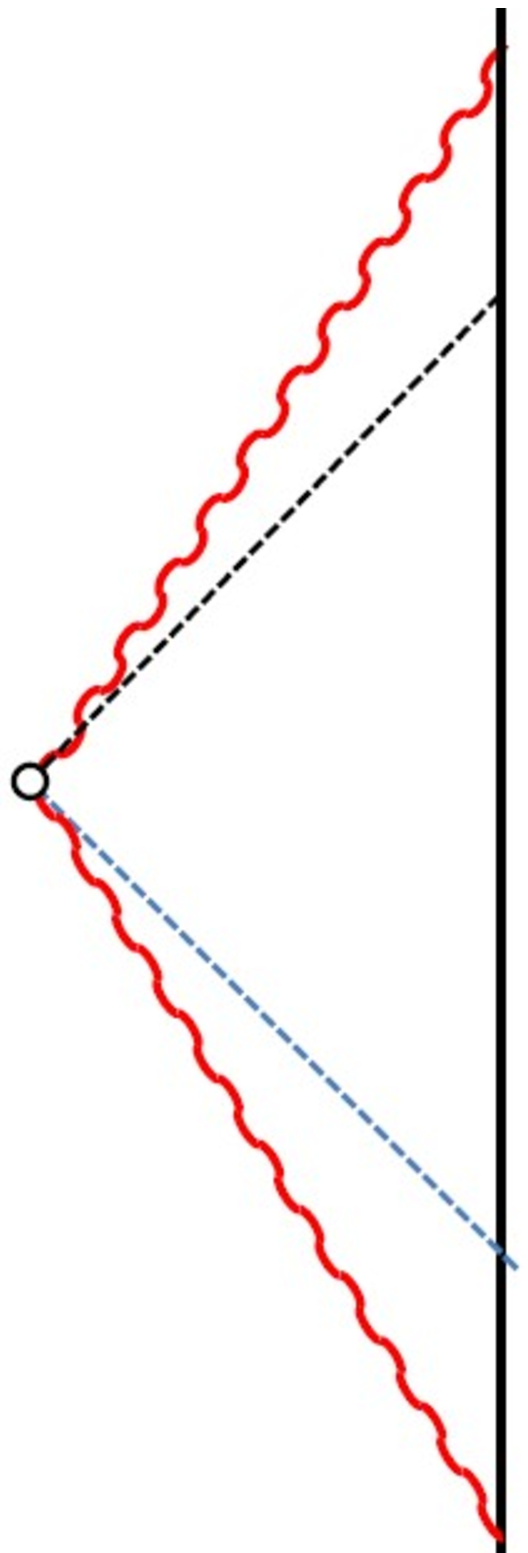}
	 		\caption{$a=1$}
	 		\label{penroseTypeIIaa1}
	 	\end{minipage}
	 \end{figure}
	 \begin{figure}[H]
	 	\begin{minipage}{0.3\hsize}
	 		\centering\includegraphics[scale=0.3]{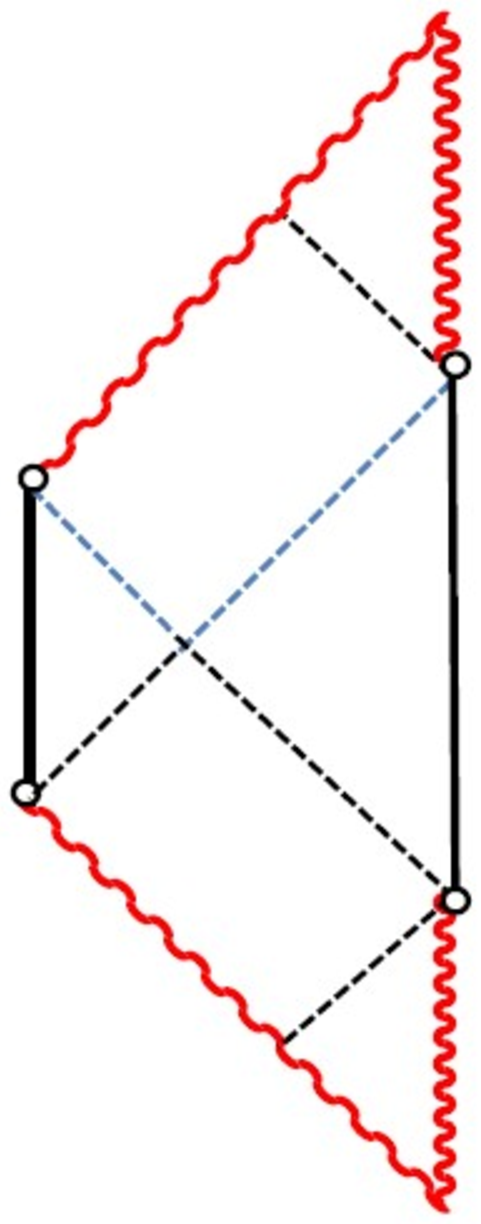}
	 		\caption{$a=\sqrt{2}$}
	 		\label{penroseTypeIIaasqrt2}
	 	\end{minipage}
	 	\begin{minipage}{0.3\hsize}
	 		\centering\includegraphics[scale=0.3]{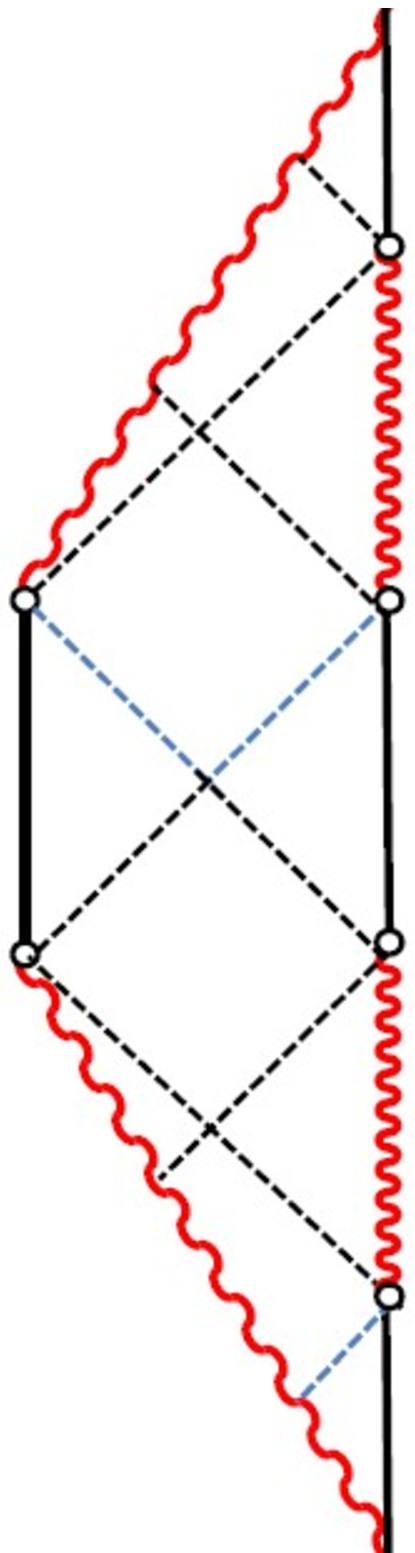}
	 		\caption{$1<a<2,\ (a\ne\sqrt{2})$}
	 		\label{penroseTypeIIaafrom1to2}
	 	\end{minipage}
	 	\begin{minipage}{0.3\hsize}
	 		\centering\includegraphics[scale=0.3]{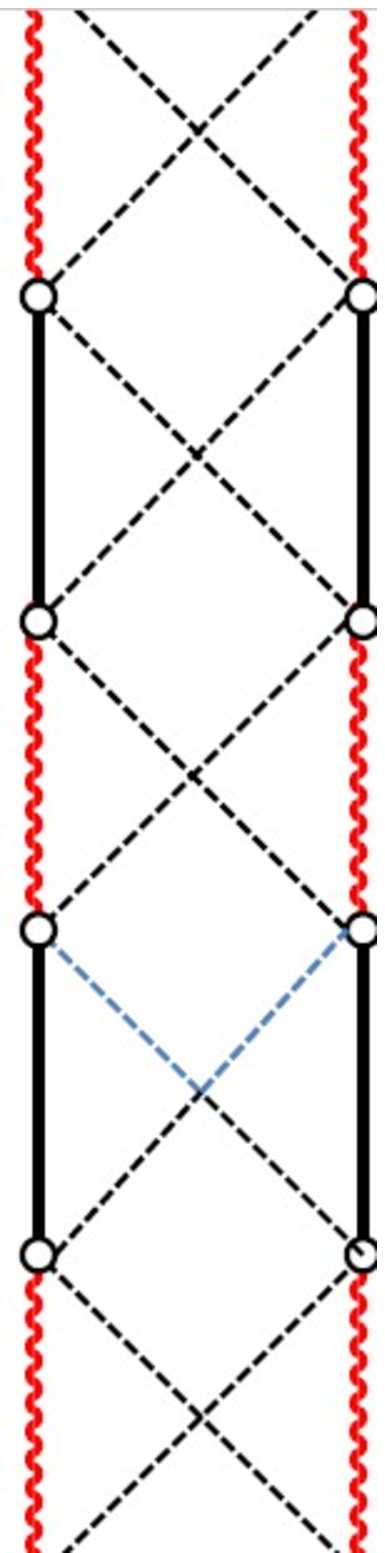}
	 		\caption{$2<a$}
	 		\label{penroseTypeIIaafrom2to}
	 	\end{minipage}
	 \end{figure}
	 In the case of $2<a$ the patterns are repeated vertically like the Reissner-Nordstr\"om black hole. 
	 There exist event horizons and inner horizons. 
	 The causal diagram of this spacetime is the same as type $I_b$ with $a\ne0$ case, 
	 however the local geometries of these two spacetime is different. 
	 In the case of $1<a<2,\ (a\ne\sqrt{2})$  null infinity of right hand side is devided into three pieces by 
	 the singularities. The timelike singularities bound the spacetime in the finite region of the conformal 
	 diagram. The spacetime has event horizons and inner horizons.
	 The case $a=\sqrt{2}$ is limiting case of the previous one 
	 such that the timelike singularities become null singularities.   
	 The $a=1$ case is also the limit of the $1<a<2$ case such that the null infinity of right hand side 
	 shrinks to a point.

	 \subsection{Type$I_a:\ b(L_{ts}-L_{xy})+a(K_{ty}+K_{sx})$}
	 The representation matrix is
	 \begin{align}
	 \xi=\left(\begin{matrix}
	 0&-b&0&a\\
	 b&0&a&0\\
	 0&a&0&b\\
	 a&0&-b&0
	 \end{matrix}\right)
	 \end{align}
	 then 
	 \begin{align}
	 \phi=a^2-b^2-4ab(tx-sy)
	 \end{align}
	 then 
	 \begin{align}
	 d\phi=-4ab(xdt+tdx-yds-sdy),
	 \end{align}
	 hence
	 \begin{align}
	 d\phi(X)|_N=-4ab\sin(\tau-\theta).
	 \label{type Ia dphi}
	 \end{align}
	 
	 We suppose $\phi(N)>0$, that is $\phi(N)=a^2-b^2>0$.
	 
	 The orthogonal condition of $\xi N$ and $X$ is
	 \begin{align}
	 \eta(\xi^\ast,X)_N=b+a\cos(\tau-\theta)=0.
	 \label{type Ia orthogonal condi}
	 \end{align}
	 
	 Since eq\eqref{type Ia orthogonal condi} and eq\eqref{type Ia dphi} we have
	 \begin{align}
	 {\rm sign}(\lambda_0^\pm)={\rm sign}(-d\phi(X)_N)={\rm sign}(\mp4b\sqrt{a^2-b^2}).
	 \end{align}
	 The signs of $\lambda^\pm_0$ are fixed for all $\tau$.
	 
	 Therefore the conformal diagram is shown in Figure \ref{Ia}.
	 \begin{figure}[H]
	 	\centering
	 	\includegraphics[keepaspectratio, scale=0.6]{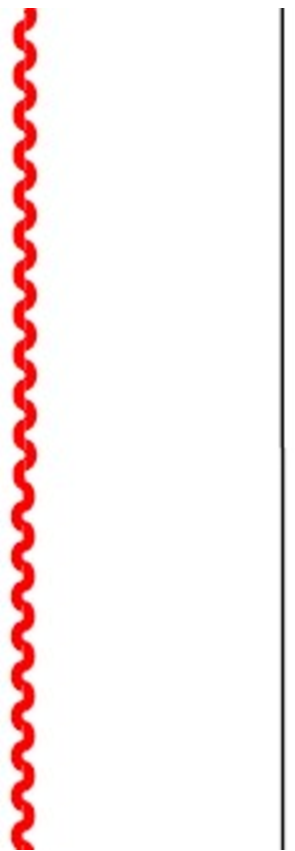}
	 	\caption{Conformal diagram of the Type$I_a$}
	 	\label{Ia}
	 \end{figure}
	 
	 This spacetime dose not have event horizon.
	 
	 \subsection{Type$III^+:\ K_{tx}+L_{xy}$}
	 The representation matrix is
	 \begin{align}
	 \xi=\left(\begin{matrix}
	 0&0&1&0\\
	 0&0&0&0\\
	 1&0&0&-1\\
	 0&0&1&0
	 \end{matrix}\right)
	 \end{align}
	 then 
	 \begin{align}
	 \phi=\frac{1}{2}(t-y)^2
	 \end{align}
	 and 
	 \begin{align}
	 d\phi=(t-y)(dt-dy).
	 \end{align}
	 So we obtain
	 \begin{align}
	 d\phi(X)|_N=\sin\tau(\cos\tau-\sin\theta_0).
	 \end{align}
	 
	 And since 
	 \begin{align}
	 \phi(N)=\sin^2\tau
	 \end{align}
	 therefore we may consider in $0<\tau<\pi$.
	 
	 The orthogonal condition of $\xi^\ast_N(=\xi N)$ and $X$ is
	 \begin{align}
	 \cos\theta_0\sin\tau=0\Leftrightarrow \theta_0^\pm=\pm\pi/2\nonumber,
	 \end{align}
	 then ${\rm sign}(\lambda_0^\pm)= {\rm sign}(-d\phi(X^\pm)_N)={\rm sign}(\sin\tau(\cos\tau\pm1))$ holds and graph is follows.
	 
	 \begin{figure}[H]
	 	\centering
	 	\includegraphics[keepaspectratio, scale=0.3]{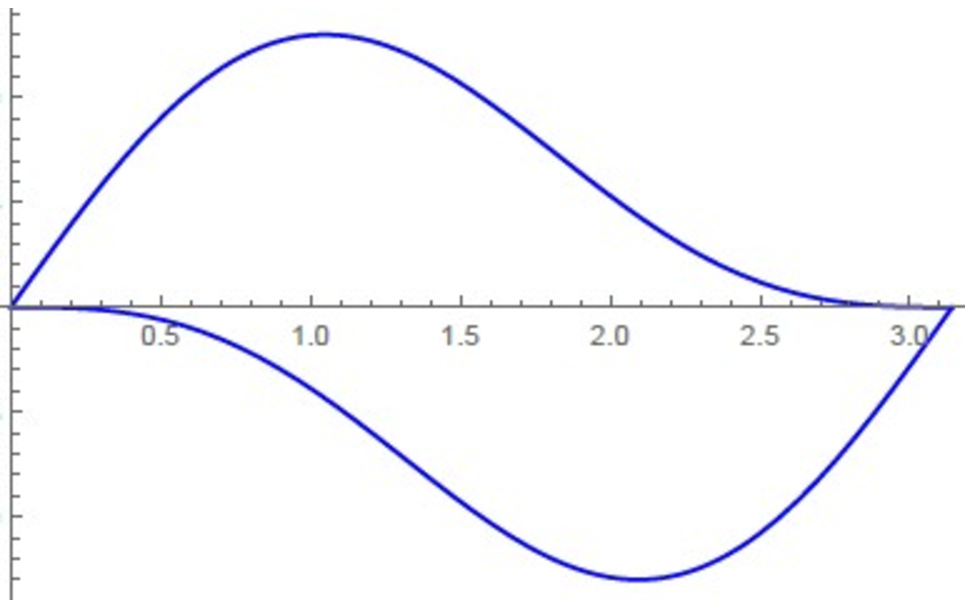}
	 	\caption{}
	 	\label{}
	 \end{figure}
	 
	 And since
	 \begin{align}
	 |\xi X|^2=(\cos\tau-\sin\theta)^2,
	 \end{align}
	 then
	 \begin{align}
	 |\xi X^\pm|^2=(\cos\tau\mp1)^2.
	 \end{align}

	 In the case of $\tau=0$, the null geodesic directed $\theta_0^+$ is always on the singularity because $|\xi X^+|^2=0,\ |\xi N|^2=0$ then $f^+(\lambda)=|\xi X^+|^2\lambda^2+|\xi N|^2=0$ holds.
	 Therefore the singular is null.
	 Similarly, in the case of $\tau=\pi$, null geodesic directed $\theta_0^-$ is always on the singularity.
	 
	 We obtain the following conformal diagram \ref{penroseK_tx+L_xy}.
	 \begin{figure}[H]
	 	\centering\includegraphics[scale=0.3]{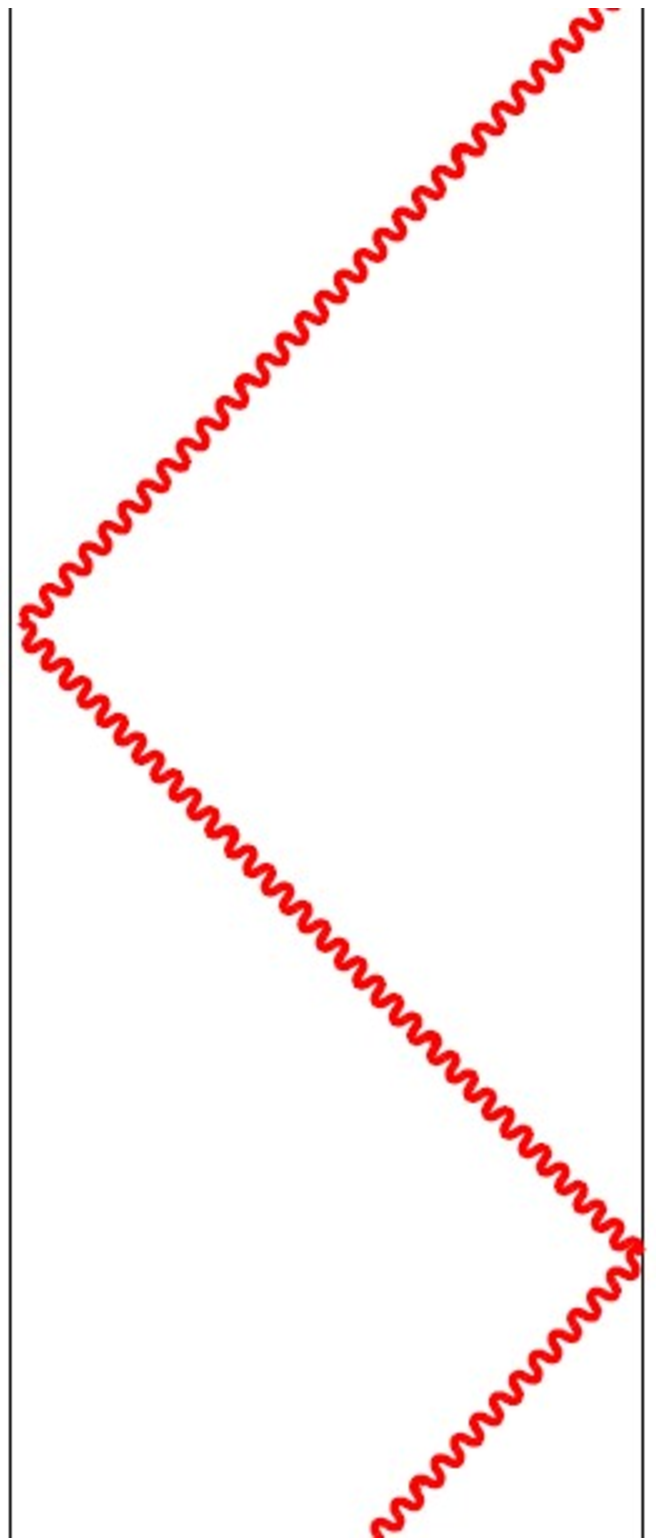}
	 	\caption{}
	 	\label{penroseK_tx+L_xy}
	 \end{figure}
	 
	 \section{Causality of base space in AdS$_5$}
	 \label{AdS_5}
	 In this section we investigate the  base space in AdS$_5$.
	 We choose a simple and interesting Killing vector of AdS$_5$, $\xi=aK_{tx}+bK_{sy}+cL_{zw},\ (a\ne0,b\ne0)$.
	 The conclusion of this section holds in the case of both $c=0$ and $c\ne0$.
	 
	 Let $\{t,s,x,y,z,w\}$ be a coordinate system of six-dimensional pseudo Euclidean space $E^{(2,4)}$. 	
	 We set a coordinate system of AdS$_5$ as
	 \begin{align}
	 t&=\sqrt{r^2+1}\sin\tau,\\
	 s&=\sqrt{r^2+1}\cos\tau,\\
	 x&=r\cos\rho\sin\theta,\\
	 y&=r\sin\rho,\\
	 z&=r\cos\rho\cos\theta\cos\phi,\\
	 w&=r\cos\rho\cos\theta\sin\phi.
	 \end{align}

	 Let
	 \begin{align}
	 N={}^T(\sin\tau,\cos\tau,0,0,0,0), 
	 \end{align}
	 then a null vector on $N$ is given by
	 \begin{align}
	 X={}^T(\cos\tau,-\sin\tau,\cos\rho_0\sin\theta_0,\sin\rho_0,\cos\rho_0\cos\theta_0\cos\phi_0,\cos\rho_0\cos\theta_0\sin\phi_0)
	 \label{XAdS5}
	 \end{align}
	 where $\theta_0,\phi_0,\rho_0$ are direction parameters.
	 
	 The representation matrix of $\xi$ is
	 \begin{align}
	 \xi=\left(
	 \begin{matrix}
	 0&0&a&0&0&0\\
	 0&0&0&b&0&0\\
	 a&0&0&0&0&0\\
	 0&b&0&0&0&0\\
	 0&0&0&0&0&-c\\
	 0&0&0&0&c&0
	 \end{matrix}
	 \right)
	 \end{align}
	 
	 We have
	 \begin{align}
	 \xi X&={}^T(a \cos\rho_0 \sin\theta_0, b \sin\rho_0, a \cos\tau, -b \sin\tau, -c \cos\rho_0\cos\theta_0\sin\phi_0, c \cos\phi_0\cos\rho_0 \cos\theta_0),\\
	 \xi N&={}^T(0,0,a \sin\tau,b\cos\tau,0,0),
	 \end{align}
	 and the orthogonal condition is given by
	 \begin{align}
	 \eta(\xi N,X)&=a\sin\tau\cos\rho_0\sin\theta_0+b\cos\tau\sin\rho_0=0.
	 \label{AdS_5rho_0}
	 \end{align}
	 
	 We want to search for null geodesics which intersect with the singularity, so we will find $X$ such that $H(\xi X,\xi N)=0$ from proposition\ref{main prop}.
	 
	 After some calculation we obtain 
	 \begin{align}
	 H(\xi X,\xi N):&=\frac{b^2g(\tau)\cos^2\theta_0}{2(b^2+a^2\sin^2\theta_0\tan^2\tau)},\label{AdS_5H}
	 \end{align}
	 where $g(\tau)=b^2c^2+a^2(2b^2+c^2)-(a^2-b^2)c^2\cos 2\tau$.
	 
	 The remarkable difference that $H$ vanishes identically in AdS$_3$ while in AdS$_5$ H depends on $\tau$ and direction parameters.
	 Hence we will find zero point of H to find a intersection with the singularity.
	 
	 From \eqref{AdS_5H}, $H=0$ requires $g(\tau)=0$ or $\theta_0=\pi/2$.
	 If $g(\tau_0)=0$ then for $a\ne0,b\ne0,c\ne0$
	 \begin{align}
	 |\cos 2\tau|=\left|\frac{b^2c^2+a^2(2b^2+c^2)}{(a^2-b^2)c^2}\right|>1\nonumber
	 \end{align}
	 holds.
	 Therefore there exists no $\tau_0$ such that $g(\tau_0)=0$.
	 So it implies that $H=0$ if and only if $\theta_0=\pi/2$.
	 
	 The norm of $\xi$ is given by
	 \begin{align}
	 \phi=\frac{1}{2}|\xi|^2&=\frac{1}{2}(-a^2x^2-b^2y^2+a^2t^2+b^2s^2+c^2z^2+c^2w^2),
	 \end{align}
	 then we have
	 \begin{align}
	 d\phi&=-a^2xdx-b^2ydy+a^2tdt+b^2sds+c^2zdz+c^2wdw,
	 \end{align}
	 therefore
	 \begin{align}
	 {\rm sign}(\lambda_0)={\rm sign}(-d\phi(X)|_{N})={\rm sign}(-(a^2-b^2)\sin2\tau).
	 \end{align} 
	 
	 The change of the sign of $\lambda_0$ as increasing $\tau$ for $\theta_0=\pi/2$ is shown in the figure \ref{typeIblambda0sign}.
	 
	 In the case of $\tau=\frac{1}{2}(2n+1)\pi$ then $f(\lambda)=a^2$ and in the case of $\tau=n\pi$ then $f(\lambda)=b^2$ holds.
	 This is the same as the case of $a\ne0$ in type $I_b$.
	 Therefore the conformal diagram is given in figure \ref{penroseaK_tx+bK_sy+cL_zwabcne0}.
	 
	 \begin{figure}[H]
	 	\centering\includegraphics[scale=0.4]{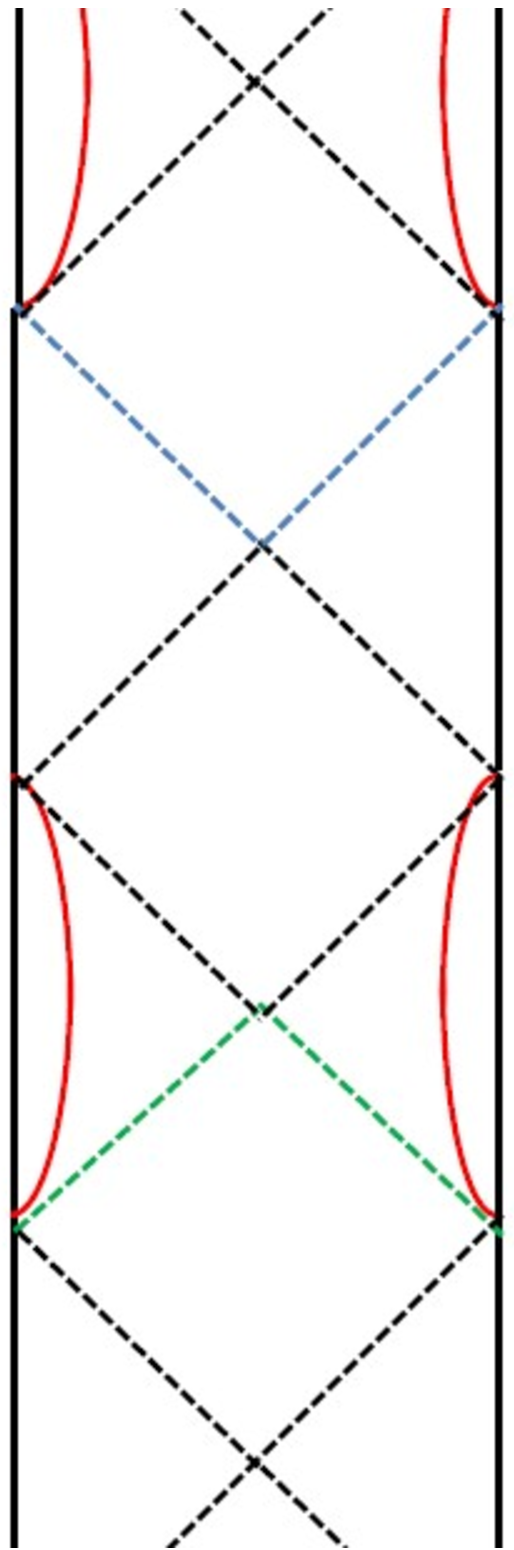}
	 	\caption{}
	 	\label{penroseaK_tx+bK_sy+cL_zwabcne0}
	 \end{figure}
	 
	 In the Figure\ref{penroseaK_tx+bK_sy+cL_zwabcne0}, the wavy lines (red) are singularity and the dotted lines (blue and green) are outer horizon and inner horizon respectively if we restrict the spacetime $\theta_0=\pi/2$.
	 For all null geodesics directed $\theta_0\ne\pi/2$, we have $H\ne0$, that is, the geodesics are extendible to the both of the future and the past null infinity. 
	 And since in this case null infinity is timelike and connected, then there exists no event horizon in the whole spacetime.
	 Therefore it is not a black hole.
	 However, if we restrict the spacetime to the section of $\theta_0=\pi/2$ then the two-dimensional submanifold has black hole structure.
	 This spacetime represents that two-dimensional black hole is embedded in the four-dimensional spacetime.
	 
	 \section{Conclusion and discussion}\label{section Conclusion and discussion}
	 
	 The geometrical structure and the causal structur of the basespaces of Riemannian submersion of AdS 
	 have a rich variety  even if AdS has simple structure. 
	 We have shown explicitely that the causal structure of all base space of Riemannian submersion 
	 by one-dimensional isometry group of AdS$_3$ by using the fact that AdS is embedded in a  flat space. 
	 
	 The bases paces which have black hole structure in AdS$_3$ are induced by two types of Killing vectors: 
	 type $I_b$, $\xi=K_{tx}+aK_{sy}$, and type $II_a$, $\xi=a(K_{tx}+K_{sy})+L_{ts}-K_{ty}-K_{sx}-L_{xy}$.
	 In the case of type $I_b$, 
	 we obtain static black hole for $a=0$, while we obtain Reissner-Nordstr\"om like black hole for $a\ne0$. 
	 In the case of type $II_a$,
	 if $a<1$ no black hole structure appears, and if $1\leq a$ there exist four kind of black holes. 
	 There are a variety of causal structures that express the black holes.

	 In the case of AdS$_5$ we have observed one type of static black hole. 
	 However, there exists an interesting four-dimensional spacetime which has two-dimensional black hole 
	 as a submanifold.
	 Furthermore, the static black hole appears in general higher dimension as shown 
	 in Appendix \ref{Causality in norm-weighted Orbit Space of AdS_n}.
	 
	 Finally, we found that if the base space has an event horizon or an inner horizon then $|\xi|^2$ is constant on the horizon, where $\xi$ is a Killing vector field used in submersion. 
	 It means that the null vector that generates the horizon is an eigen vector of $\xi$ on the horizon.
	 Furthermore, the horizon is the Killing horizon	because there exists a Killing vector which commutes with $\xi^\ast$ tangents to the horizon (see Appendix \ref{Killing horizon of norm-weighted orbit space}).
	 
	 We could clarify the causal structure of the base space in the case of dimension of isometry group is one by using the algebraic method.
	 It is interesting that we generalize the case of over two-dimensional isometry group.
	 We will report this issue in the future.
	 
	 \section*{Acknowledgements}
	 The authors thank K. Nakao, M.Morisawa and T. Houri for useful
	 discussions. This work was supported by JSPS KAKENHI Grant Number JP16K05358 (HI).
	 
	 \section*{Appendix}
	 \appendix
	 
	 \section{A brief review of the Kaluza-Klein theory and the Einstein-Maxwell-Dilaton theory}\label{A brief review of the Kaluza-Klein theory and the Einstein-Maxwell-Dilaton theory}

	 Let $(\tilde{M},\tilde{g})$ be an $n+1$ dimensional Lorentzian manifold, and $\xi$ be a Killing vector field generate a one dimensional isometry group.
	 Assume that $\{x^1,\cdots,x^n,x^{n+1}=y\}$ is coordinate system of $\tilde{M}$ and $\partial_{y}=\xi$.
	 Then we decompose $\tilde{g}$ into
	 \begin{align}
	 \tilde{g}&=g_{ij}dx^idx^j+e^{2\phi}(dy+A_\mu dx^\mu)^2, \quad(1\le i,j\le n,1\le\mu\le n+1),
	 \end{align}
	 where
	 \begin{align}
	 e^{2\phi}=|\xi|^2,\ A_\mu=\tilde{g}_{\mu\nu}\xi^\nu.
	 \end{align}
	 
	 Einstein-Hilbelt's action in $(\tilde{M},\tilde{g})$ is decomposed as
	 \begin{align}
	 S&=\int d^{n+1}\boldsymbol{x}\sqrt{-|\tilde{g}|}\tilde{R}\\
	 &=\int d^n\boldsymbol{x}dy\sqrt{-|g|}e^\phi(R-\frac{1}{4}e^{2\phi} F^2-e^{-\phi}\Delta e^\phi).
	 \label{reduction action}
	 \end{align}

	 Moreover we apply the conformal transformation to $g$
	 \begin{align}
	 \hat{g}_{ij}=e^{\alpha\phi}g_{ij},
	 \label{einstein frame metric}
	 \end{align}
	 for a suitable constant $\alpha$.

	 If we tune $\alpha$ as
	 \begin{align}
	 -(n-2)\alpha+1=0,
	 \label{einstein frame alpha}
	 \end{align}
	 then the action is given by
	 \begin{align}
	 S=\ const.\int d^{n-1}\boldsymbol{x}\sqrt{-|\hat{g}|}\hat{R}+\ other\ term.
	 \end{align}
	 This action gives theory of Einstein gravity and matters which interact with gravity minimally.
	 This is called the Einstein frame.
	 We call $(M,\hat{g})$ base space.
	 If we apply this procedure to the spacetime with cosmological constant then we obtain the Einstein-Maxwell-Dilaton theory with Liouville potential.
	 
	 \section{Causality of base space in AdS$_n$}\label{Causality in norm-weighted Orbit Space of AdS_n}
	 We show that the specific base space of AdS$_n$ has black hole structure.
	 
	 Let $\{t,s,x^1,x^2,\cdots,x^{n-1}\}$ be a coordinate system of 6 dimensional pseudo Euclidean space $E^{(2,n-1)}$ and its metric is given by $\eta=-dt^2-ds^2+(dx^1)^2+\cdots+(dx^{n-1})^2$ and AdS$_n$ is defined by $-t^2-s^2+(x^1)^2+\cdots+(x^{n-1})^2=-1$.
	 We consider the base space by the Killing vector $\xi=K_{tx^1}$.
	 
	 Let
	 \begin{align}
	 N:={}^t(\sin\tau,\cos\tau,0,\cdots,0)
	 \end{align}
	 then a null vector on $N$ is given by
	 \begin{align}
	 X&=\cos\tau\partial_t-\sin\tau\partial_s+\Omega_{n-1}
	 \label{X in AdSn}
	 \end{align}
	 where $\Omega_{n-1}=\Omega^i\partial_{x^i}$ is a unit vector.
	 
	 And the represent matrix of $\xi$ is given by
	 \begin{align}
	 \xi=\left(\begin{matrix}
	 0&0&1&0&\cdots&0\\
	 0&0&0&0&\cdots&0\\
	 1&0&0&0&\cdots&0\\
	 \cdots\\
	 \cdots\\
	 0&0&0&0&\cdots&0
	 \end{matrix}\right)
	 \end{align}
	 hence
	 \begin{align}
	 \xi X&={}^t(\Omega^1,0,\cos\tau,0,\cdots,0)\\
	 \xi N&={}^t(0,0,\sin\tau,0,\cdots,0)
	 \end{align}
	 
	 Orthogonal condition of $X$ and $\xi$ is
	 \begin{align}
	 \eta(X,\xi)|_N=\eta(X,\sin\tau\partial_{x^1})=\Omega^1\sin\tau=0
	 \end{align}
	 
	 Therefore we obtain for arbitrary $\tau$
	 \begin{align}
	 X\wedge\xi X\wedge\xi N=0
	 \end{align}
	 
	 Next we observe a sign of $-d\phi(X)|_N$, $|\xi X|^2$ and $f(0)=|\xi^\ast|^2_N$, where $\phi=\frac{1}{2}(-(x^1)^2+t^2)$.
	 
	 \begin{align}
	 -d\phi(X)=(x^1dx^1-tdt)(X)
	 \end{align}
	 then
	 \begin{align}
	 -d\phi(X)|_N=-\sin\tau\cos\tau=-\frac{1}{2}\sin2\tau
	 \end{align}
	 
	 Moreover
	 \begin{align}
	 |\xi X|^2&=\cos^2\tau\ (\tau\ne n\pi)\\
	 &=1\ (\tau=n\pi)
	 \end{align}
	 and
	 \begin{align}
	 f(0)=\sin^2\tau
	 \end{align}
	 
	 In the case of $\tau=\frac{\pi}{2}$, $f(\lambda)=1$ so $\lambda_0$ is not defined.
	 
	 Therefore the causal structure of this spacetime is isotropic then the conformal diagram for arbitrary direction is follows.
	 \begin{figure}[H]
	 	\centering
	 	\includegraphics[keepaspectratio, scale=0.25]{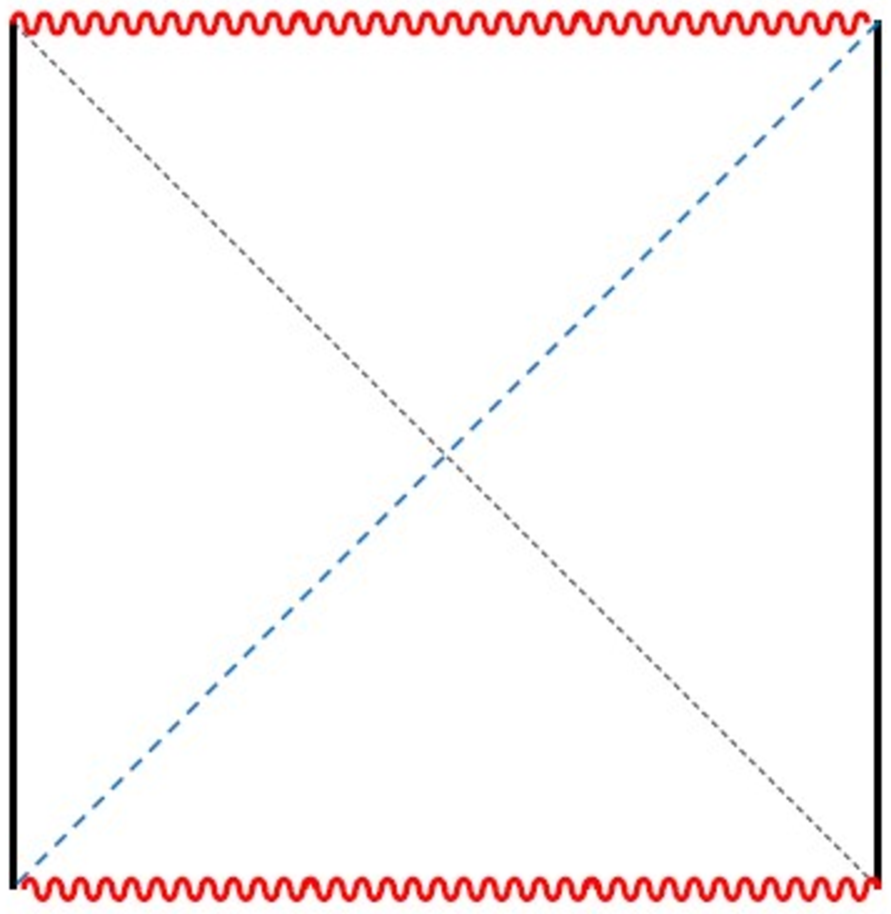}
	 	\caption{}\label{}
	 \end{figure}
	 
	 \section{Horizon of base space}\label{Killing horizon of norm-weighted orbit space}
	 We suppose the horizon given by $\pi(x(\lambda))$ where $x(\lambda)=N+X\lambda$ is a horizontal null geodesic for a suitable point $N$ and null vector $X$ and $\pi$ is the Riemannian submersion by $\xi$.
	 And let $\xi^\ast$ is a Killing vector used by Riemmanian submersion.
	 In our study it is observed that the null vector X is an eigen vector of $\xi$.
	 Then $|\xi X|^2=\eta(\xi X,\xi X)=0$ holds and by Proposition \ref{constantxi} the norm of $\xi^\ast$ is constant on the null geodesic $x(\lambda)$.
	 Moreover the horizon is Killing horizon as follows.
	 There exists a Killing vector field $\zeta^\ast$ commute with $\xi^\ast$ and $\pi_\ast\zeta^\ast$ along a curve $|\xi|^2=const$.
	 Therefore $\pi_\ast\zeta^\ast$ tangents to the horizon $\pi(x(\lambda))$, that is, the horizon is a Killing horizon.


\begin{thebibliography}{99}
	 	\bibitem{Maldacena:1997re} 
	 	J.~M.~Maldacena,
	 	``The Large N limit of superconformal field theories and supergravity,''
	 	Int.\ J.\ Theor.\ Phys.\  {\bf 38}, 1113 (1999)
	 	[Adv.\ Theor.\ Math.\ Phys.\  {\bf 2}, 231 (1998)]
	 	doi:10.1023/A:1026654312961, 10.4310/ATMP.1998.v2.n2.a1
	 	[hep-th/9711200].
	 	
	 	\bibitem{Maartens:2010ar} 
	 	R.~Maartens and K.~Koyama,
	 	``Brane-World Gravity,''
	 	Living Rev.\ Rel.\  {\bf 13}, 5 (2010)
	 	doi:10.12942/lrr-2010-5
	 	[arXiv:1004.3962 [hep-th]].
	 	
	 	\bibitem{BTZ}
	 	M. Banados, C. Teitelboim and J. Zanelli, Phys. Rev. Lett., 69, 1849 (1992)
	 	
	 	\bibitem{Banados:1992gq} 
	 	M.~Banados, M.~Henneaux, C.~Teitelboim and J.~Zanelli,
	 	``Geometry of the (2+1) black hole,''
	 	Phys.\ Rev.\ D {\bf 48}, 1506 (1993)
	 	Erratum: [Phys.\ Rev.\ D {\bf 88}, 069902 (2013)]
	 	doi:10.1103/PhysRevD.48.1506, 10.1103/PhysRevD.88.069902
	 	[gr-qc/9302012].
	 	
	 	\bibitem{Holst:1997tm} 
	 	S.~Holst and P.~Peldan,
	 	``Black holes and causal structure in anti-de Sitter isometric space-times,''
	 	Class.\ Quant.\ Grav.\  {\bf 14}, 3433 (1997)
	 	doi:10.1088/0264-9381/14/12/025
	 	[gr-qc/9705067].
	 	
	 	\bibitem{Koike:2008fs} 
	 	T.~Koike, H.~Kozaki and H.~Ishihara,
	 	``Strings in five-dimensional anti-de Sitter space with a symmetry,''
	 	Phys.\ Rev.\ D {\bf 77}, 125003 (2008)
	 	doi:10.1103/PhysRevD.77.125003
	 	[arXiv:0804.0084 [gr-qc]].
	 	
	 	
	 	\bibitem{Morisawa:2017lpj} 
	 	Y.~Morisawa, S.~Hasegawa, T.~Koike and H.~Ishihara,
	 	``Cohomogeneity-one-string integrability of spacetimes,''
	 	arXiv:1709.07659 [hep-th].
	 	
	 	
	 	\bibitem{Aminneborg:2008sa} 
	 	S.~Aminneborg and I.~Bengtsson,
	 	``Anti-de Sitter Quotients: When Are They Black Holes?,''
	 	Class.\ Quant.\ Grav.\  {\bf 25}, 095019 (2008)
	 	doi:10.1088/0264-9381/25/9/095019
	 	[arXiv:0801.3163 [gr-qc]].
	 	
	 	\bibitem{JML}
	 	John M Lee, Introduction to Smooth Manifolds, GTM(Springer)
	 	
	 	
	 	
	 	
	 	
	 \end{thebibliography}
\end{document}